\documentclass[a4paper,USenglish,dvipsnames]{article}
\usepackage[T1]{fontenc}
\usepackage[breaklinks,hidelinks]{hyperref} 
\usepackage{fullpage}
\usepackage{amssymb}
\usepackage{amsmath}
\usepackage{amsthm}
\usepackage{thmtools}
\usepackage{mathtools}
\usepackage[capitalise]{cleveref}

\usepackage[color = blue!50]{todonotes}

\usepackage{amsmath}
\usepackage{amssymb}
 \newtheorem{assumption}{Assumption}

\usepackage{bm}
\usepackage{bbm}

\usepackage{xspace}
\usepackage{turnstile}

\usepackage{booktabs}
\usepackage{multirow}
\usepackage{pifont}
\usepackage{mathtools}

\usepackage{multicol}
\usepackage{stackengine}
\usepackage{pgf}

\usepackage{tikz}
\usetikzlibrary{decorations.pathreplacing,positioning}
\usetikzlibrary{arrows,automata,calc}
\tikzset{>=stealth, shorten >=1pt}
\tikzset{every edge/.style = {thick, ->, draw}}
\tikzset{every loop/.style = {thick, ->, draw}}
\pgfdeclarelayer{background}
\pgfsetlayers{background,main}

\usetikzlibrary{
  calc,
  arrows,
  positioning,
  shapes,
  shadows,
  automata,
  external,
  shadows,
  decorations.pathreplacing,
  decorations.pathmorphing,
  decorations.markings,
  backgrounds,
  patterns
}

\tikzset{  
  node distance=4em,
  >=stealth',
  shadow/.style		= {opacity=.25, black, shadow xshift=0.08, shadow yshift=-0.08},
  plainnode/.style 	= {draw, ultra thick, fill=gray!10},
  pl0/.style	       = {circle, minimum size=6mm, inner sep=0mm,plainnode, drop shadow = {shadow}},
  ind/.style         = {circle,draw,fill=black},
  nind/.style        = {circle,draw},
  wt/.style          = {text = white, yshift = -0.1em,scale=0.9}
}

\usepackage{anyfontsize}

\newcommand{\myquote}[1]{``#1''}

\newcommand{\nats}{\mathbbm{N}}
\renewcommand{\epsilon}{\varepsilon}
\renewcommand{\phi}{\varphi}
\newcommand{\size}[1]{|#1|}

\newcommand{\set}[1]{\{#1\}}

\newcommand{\vectordots}[2]{
\begin{pmatrix}
  #1 \\
    \vspace{-.6cm}\\
  \vdots \\
    \vspace{-.6cm}\\
  #2\\
\end{pmatrix}}

\newcommand{\vectordotswithtwoargsatend}[3]{
\begin{pmatrix}
  #1 \\
    \vspace{-.5cm}\\
  \vdots \\
    \vspace{-.5cm}\\
  #2\\
  #3\\
\end{pmatrix}}

\newcommand{\lang}[3]{L_{#3}(#1,#2)}

\newcommand{\F}{\mathop{\mathbf{F}\vphantom{a}}\nolimits}
\newcommand{\G}{\mathop{\mathbf{G}\vphantom{a}}\nolimits}

\newcommand{\ltl}{{LTL}\xspace}

\newcommand{\ctlstar}{{CTL$^*$}\xspace}
\newcommand{\hyltl}{{HyperLTL}\xspace}
\newcommand{\hyctlstar}{{HyperCTL$^*$}\xspace}

\newcommand{\Q}{Q}

\newcommand{\hylogicpda}{HyperPDA\xspace}
\newcommand{\hylogicvpa}{HyperVPA\xspace}

\newcommand{\Succ}[1]{\mathrm{Succ}(#1)}

\newcommand{\var}{\mathcal{V}}

\newcommand{\gni}{\mathrm{GNI}}
\newcommand{\inp}{\mathrm{in}}
\newcommand{\out}{\mathrm{out}}
\newcommand{\propo}{\texttt{p}}

\newcommand{\pspace}{\textsc{PSpace}\xspace}
\newcommand{\expt}{\textsc{ExpTime}\xspace}

\newcommand{\twoexp}{\textsc{2ExpTime}\xspace}

\newcommand{\tower}{\textsc{Tower}\xspace}

\newcommand{\tsys}{\mathfrak{T}}
\newcommand{\traces}[1]{{L}(#1)}

\newcommand{\aut}{\mathcal{A}}
\newcommand{\autb}{\mathcal{B}}

\newcommand{\autd}{\mathcal{D}}
\newcommand{\initmark}{I}
\newcommand{\compl}[1]{(#1)^c}

\newcommand{\combine}[1]{\mathrm{mrg}(#1)}
\newcommand{\projone}[1]{\mathrm{pr}_V(#1)}
\newcommand{\projtwo}[1]{\mathrm{pr}_{\cal{P}}(#1)}

\newcommand{\Gammabot}{\Gamma_{\!\bot}}
\newcommand{\sh}{\mathrm{sh}}
\newcommand{\trans}[1]{\xrightarrow{#1}}

\newcommand{\Sigmacall}{\Sigma_c}
\newcommand{\Sigmaskip}{\Sigma_s}
\newcommand{\Sigmareturn}{\Sigma_r}

\newcommand{\gsgame}{\mathcal{G}}
\newcommand{\SigmaI}{\Sigma_1}
\newcommand{\SigmaO}{\Sigma_2}

\newcommand{\step}{\texttt{Step}}
\newcommand{\properstep}{\texttt{ProperStep}}

\newcommand{\effect}{\mathit{ef}}
\newcommand{\ints}{\mathbbm{Z}}

\newcommand{\prophyCall}[2]{\texttt{CallHump}\langle{#1},{#2}\rangle}
\newcommand{\prophySkip}[2]{\texttt{SkipHump}\langle{#1},{#2}\rangle}
\newcommand{\prophyReturn}[2]{\texttt{ReturnHump}\langle{#1},{#2}\rangle}

\newcommand{\prophyMCallAcc}[2]{\texttt{MatchedCallStep}\langle{#1},{#2}\rangle}
\newcommand{\prophyUCallAcc}[2]{\texttt{UnmatchedCallStep}\langle{#1},{#2}\rangle}

\newcommand{\prophySkipAcc}[2]{\texttt{SkipStep}\langle{#1},{#2}\rangle}
\newcommand{\prophyReturnAcc}[2]{\texttt{UnmatchedReturnStep}\langle{#1},{#2}\rangle}

\newcommand{\firstF}[1]{{\mathrm{FirstAcc}}({#1})}
\newcommand{\type}[1]{{\mathit{type}}({#1})}
\newcommand{\ch}{\mathit{choice}}

\newcommand{\opt}[1]{\mathrm{opt}\langle #1 \rangle}
\newcommand{\nonempty}{\textsc{nonempty}}

\newcommand{\pvar}[1]{x_{#1}}

\newtheorem{theorem}{Theorem}
\newtheorem{lemma}{Lemma}
\newtheorem{corollary}{Corollary}
\newtheorem{proposition}{Proposition}
\newtheorem{definition}{Definition}
\newtheorem{remark}{Remark}

\crefname{equation}{Eq.}{Eqs.}
\crefname{claim}{Claim}{Claims}
\crefname{assumption}{Assumption}{Assumptions}

\title{Logics for Context-free Hyperproperties\thanks{Partly supported by the project  `Hyperlogics: Expressiveness, Monitorability and Tools (H.-Lo)' of the Icelandic Research Fund, project no.~2612260-051. 
This work has received support under the program \myquote{Investissement d'Avenir} launched by the French Government and implemented by ANR, with the reference << ANR‐18‐IdEx‐0001 >> as part of its program << Emergence >>.}}
\author{
Dietrich Kuske (Technische Universität Ilmenau, Germany)\\
Sarah Winter (Université Paris Cité, CNRS, IRIF, Paris, France)\\ 
Martin Zimmermann (Aalborg University, Aalborg, Denmark)}
\date{}

\begin{document}

\maketitle

\begin{abstract}
    We introduce a novel logic for the specification of context-free hyperproperties, which capture, e.g., the flow of information in security-critical recursive systems. Intuitively, the logic extends  visibly pushdown automata by quantification over traces, just like HyperLTL, the most important logic for regular hyperproperties, extends LTL by quantification over traces.

We show that model-checking is decidable, provided the stack height of the visibly pushdown automaton only depends on the traces bound to the variables of the initial quantifier block. 
Complementarily, we show that model-checking is undecidable for formulas with a single quantifier alternation, if the stack behavior of the visibly pushdown automaton may depend on the second quantifier block. 
This also implies that model-checking is undecidable for all fragments with more than one quantifier alternation, provided the stack behavior of the visibly pushdown automaton may depend on a non-initial quantifier block. 

Finally, we show, using a game-based approach, how to compute explanations for the satisfaction of a formula with one quantifier alternation by a given transition system. 
\end{abstract}

\section{Introduction}

The specification and verification of security-critical systems involves reasoning about the flow of information, which requires simultaneous analysis of multiple execution traces of the system.
Clarkson and Schneider~\cite{ClarksonS10} coined the term \emph{hyperproperties} for such properties, which are formally sets of sets of traces.
A system satisfies a hyperproperty if its set of traces is an element of the hyperproperty.
This should be contrasted with classical trace properties, which are sets of traces.
A system satisfies a trace property if its set of traces is a subset of the trace property.

Temporal logics are an attractive specification language both for trace and hyperproperties.
Arguably the most important logic for trace properties is Linear Temporal Logic (\ltl)~\cite{Pnueli77} while the most important logic for hyperproperties is \hyltl~\cite{ClarksonFKMRS14}, which extends \ltl with trace quantification.
For example, the formula
\[
\phi_\gni = \forall \pi.\ \forall \pi'.\ \exists \pi''.\ \G \left(\bigwedge\nolimits_{\propo \in L_{\inp} \cup L_{\out}} \propo_\pi \leftrightarrow \propo_{\pi''}\right) \wedge \G\left( \bigwedge\nolimits_{\propo \in H_\inp} \propo_{\pi'} \leftrightarrow \propo_{\pi''} \right)
\]
expresses generalized noninterference~\cite{gni}: 
for all pairs~$\pi$ and $\pi'$ of traces there is a trace~$\pi''$ that agrees with the low-security inputs (propositions in $L_\inp$) and low-security outputs (propositions in $L_\out$) of $\pi$ and the high security inputs (propositions in $H_\inp$) of $\pi'$. 
Intuitively, it is satisfied if every input-output behavior observable by a low-security user is compatible with any sequence of high-security inputs, i.e., the low security behavior does not leak information about the high-security inputs.
Many information-flow properties from the literature, e.g., generalized noninterference,  can be expressed with a single quantifier alternation~\cite{ClarksonFKMRS14}.

Model-checking of finite-state systems against \ltl and \hyltl specifications is $\pspace$-complete~\cite{sc85} and $\tower$-complete~\cite{Rabe16diss,MZ20}, respectively.
However, finite-state systems are rather restrictive, as they, for example, cannot model the call-stack of a recursive program.
Similarly, both \ltl and \hyltl are restricted to $\omega$-regular properties. 
Hence, it is natural to investigate whether more general system models and/or more expressive specification languages retain a decidable model-checking problem.

Pushdown systems (pushdown automata (PDA) without an acceptance condition) and context-free languages play a central role in these extensions: Pushdown systems naturally model recursive systems with finite data (and induce in general infinite, but finitely represented, configuration graphs) while context-free languages generalize $\omega$-regular languages with, e.g., abilities to reason about the evolution of the call stack of a program. 
Model-checking pushdown systems against \ltl is \expt-complete~\cite{DBLP:conf/concur/BouajjaniEM97} while the undecidable universality problem for context-free languages can easily be reduced to model-checking finite-state systems against context-free specifications.

However, by considering fragments of context-free languages, decidability can be regained.
Probably the most important fragment here are the visibly pushdown languages~\cite{DBLP:conf/stoc/AlurM04} (see also~\cite{m80} for earlier work on input-driven pushdown automata), where the input letter being processed determines how the stack height of a pushdown automaton evolves. 
This allows to synchronize runs of different visibly pushdown automata (VPA) on the same word and yields, e.g., much better closure properties and better algorithmic properties.
Model-checking visibly pushdown systems against visibly pushdown specifications is \expt-complete~\cite{DBLP:conf/stoc/AlurM04}.

After the introduction of \hyltl, several generalizations have been considered:
\begin{itemize}
    \item Pommellet and Touili showed that model-checking pushdown systems and visibly pushdown systems against \hyltl specifications is undecidable, even for formulas with $\exists\exists$ quantifier prefix~\cite{pt}. On the other hand, they present incomplete methods based on over- and under-approximations and show how these can be used to check security policies.

    \item Frenkel and Sheinvald introduced hypergrammars~\cite{fs}, which extend context-free grammars (over finite words) with trace quantification, just like \hyltl extends \ltl with trace quantification. Their membership problem can be seen as model-checking a regular language against a context-free hyperproperty. It can be solved in exponential time for formulas with $\exists^*$ quantifier prefix, but is undecidable for $\forall^*$ formulas.

    \item \hyctlstar extends the branching-time logic \ctlstar by trace quantification and contains \hyltl. Model-checking finite-state systems against \hyctlstar is also \tower-complete~\cite{ClarksonFKMRS14}.
    Bajwa et al.\ introduced stack-aware \hyctlstar~\cite{bajwa}, in which \hyctlstar formulas are only allowed to relate traces which have the same call-stack access pattern. Thus, all traces under consideration when evaluating a formula are synchronized. Model-checking pushdown systems against stack-aware \hyctlstar is also \tower-complete~\cite{bajwa}.

    \item Gutsfeld et al.\ presented mumbling $H_\mu$~\cite{hmulogic}, a hyperlogic for asynchronous hyperproperties based on the linear-time $\mu$-calculus and show that, under synchronization assumptions, model-checking of visibly pushdown systems against mumbling $H_\mu$ is decidable.
\end{itemize}

\subparagraph*{Our Contribution.} 
Inspired by the work of Frenkel and Sheinvald, we introduce \hylogicpda and \hylogicvpa, which extend $\omega$-PDA and $\omega$-VPA by quantification over (infinite) traces, which yields very natural logics for context-free hyperproperties. 
Thus, the formula~$\forall \pi_0.\ \exists \pi_1.\ \aut$, where $\aut$ is an $\omega$-PDA over $\Sigma \times \Sigma$ is satisfied by a system~$\tsys$ if for every trace~$t_0$ of $\tsys$, there is a trace~$t_1$ of $\tsys$ such that the pair $(t_0, t_1)$ is in $L(\aut)$.
For example, stack-aware noninference and stack-aware observational determinism~\cite{bajwa} can be expressed in \hylogicvpa, as the trace quantifiers range only over traces whose stack behavior is synchronized, which can be captured by $\omega$-VPA.

In its full generality, model-checking finite-state systems against \hylogicpda specifications is undecidable (even for formulas with a single universal quantifier), as one can easily capture universality of $\omega$-PDA.
 On the positive side, the $\exists^*$ fragment can be model-checked in exponential time, as it can be reduced to the nonemptiness problem for $\omega$-PDA.

These results highlight once again that general pushdown automata are too expressive in the context of model-checking hyperproperties. 
Thus, our main focus is on \hylogicvpa.
We prove that model-checking finite-state systems against \hylogicvpa formulas is decidable, if the stack height only depends on the letters of the traces quantified in the first quantifier block, but is undecidable if the stack-height depends on the letters of the traces quantified in a non-initial quantifier block.
Thus, we exhibit a complete picture of the decidability border for visibly context-free hyperproperties. 
Our decidability result is proven by extending the automata-based \hyltl model-checking approach~\cite{ClarksonFKMRS14,FinkbeinerRS15} to context-free specifications, which crucially relies on the fact that the stack-height is controlled by the first quantifier block, as this allows us to project tuples of traces while preserving visiblyness. 
Note that our decidable classes contain in particular stack-aware noninference and stack-aware observational determinism~\cite{bajwa}.

Finally, we consider the problem of finding explanations for why a transition system~$\tsys$ satisfies a formula or a counterexample if it does not. 
To this end, we extend the game-based characterization of $\forall^* \exists^*$~\hyltl model-checking using prophecies~\cite{BF} (see also~\cite{hyperliveness,prophy,WZtracy}) to context-free specifications.
In our setting, we employ prophecies recognized by $\omega$-VPA and construct a visibly pushdown game, which can be solved effectively~\cite{lms}. 
The winning strategies of this game, which can be implemented by pushdown transducers with lookahead, yield the desired explanations and counterexamples.

\section{Preliminaries}
\label{sec_prelims}

We denote the set of nonnegative integers by $\nats$. 

\subparagraph*{Traces and Transition Systems.}
An alphabet is a nonempty finite set. 
The sets of finite and infinite words over an alphabet~$\Sigma$ are denoted by $\Sigma^*$ and $\Sigma^\omega$, respectively. The length of a finite or infinite word~$w$ is denoted by $\size{w} \in \nats \cup \set{\infty}$. 
For a word~$w$ of length at least $n$ and $i,i'$ with $0 \le i \le i' < n$, we write $w[i,i']$ for the infix of $w$ starting at position $i$ and ending at position~$i'$ (both included).
 Given $k$ infinite words~$w_0,\ldots, w_{k-1}$, let their merge (also known as zip), which is an infinite word over $\Sigma^k$, be defined as
\[
\combine{w_0, \ldots, w_{k-1}}  = \vectordots{w_0(0)}{w_{k-1}(0)}\vectordots{w_0(1)}{w_{k-1}(1)}\vectordots{w_0(2)}{w_{k-1}(2)}\cdots .
\]
We define $\combine{w_0, \ldots, w_{k-1}}$ for finite words~$w_0, \ldots, w_{k-1}$ of the same length analogously.


A \emph{transition system}~$\tsys = (V,E,V_\initmark, \lambda)$ consists of a finite set~$V$ of vertices, a set~$E \subseteq V \times V$ of (directed) edges, a nonempty set~$V_\initmark \subseteq V$ of initial vertices, and a labelling~$\lambda\colon V \rightarrow \Sigma$ of the vertices by labels from some alphabet~$\Sigma$.
We assume that every vertex has at least one outgoing edge.
For $v \in V$, we denote by $\Succ{v}$ the set of its successors.
A \emph{path}~$\rho$ in $\tsys$ is an infinite sequence~$\rho = v_0v_1v_2\cdots$ of vertices with  $v_0 \in V_\initmark$ and $v_{n+1} \in \Succ{v_n}$ for every $n \ge 0$.
Every path~$\rho$ induces its \emph{trace}, the $\omega$-word~$ \lambda(\rho ) = \lambda(v_0)\lambda(v_1)\lambda(v_2)\cdots \in \Sigma^\omega$.
The \emph{language} of $\tsys$ is $\traces{\tsys} = \set{\lambda(\rho) \mid \rho \text{ is a path of $\tsys$}}$.
For a nonempty $V' \subseteq V$, we write $\tsys_{V'}$ to denote the transition system~$(V,E,V', \lambda)$ obtained from $\tsys$ by making $V'$ the set of initial states, and use $\tsys_v$ as shorthand for $\tsys_{\set{v}}$ for $v \in V$.


\subparagraph*{Pushdown Automata.}
An \emph{$\omega$-pushdown automaton} ($\omega$-PDA for short)~$\aut = (Q, \Sigma, \Gamma, q_\initmark, \Delta, F)$ consists of a finite set~$Q$ of states with the initial state~$q_\initmark \in Q$, an input alphabet~$\Sigma$, a stack alphabet~$\Gamma$, a transition relation~$\Delta$ to be specified, and a set~$F \subseteq Q$ of accepting states.
For notational convenience, we define $\Sigma_\epsilon = \Sigma \cup \set{\epsilon}$ and $\Gammabot = \Gamma \cup \set{\bot}$, where $\bot \notin \Gamma$ is a designated stack bottom symbol.
Then, the transition relation~$\Delta$ is a subset of  $Q \times \Gammabot \times \Sigma_\epsilon \times Q \times \Gammabot^{\le 2}$ that we require to neither write nor delete the stack bottom symbol from the stack:
If
$(q, \bot, a, q', \gamma) \in \Delta$, then $\gamma \in \bot \cdot (\Gamma \cup \set{\epsilon}) $, and if $(q, X, a, q', \gamma) \in \Delta$ for $X \in \Gamma$, then $\gamma \in \Gamma^{\le 2}$. 
Given a transition~$\tau = (q,X,a,q',\gamma)$ let $\ell(\tau) = a \in \Sigma_\epsilon$. 
We say that $\tau$ is an $\ell(\tau)$-transition and that $\tau$ is a $\Sigma$-transition, if $\ell(\tau) \in \Sigma$.
For a finite or infinite sequence~$\rho$ over $\Delta$, $\ell(\rho)$ is defined by applying $\ell$ homomorphically to every transition.

A stack content is a finite word in $\bot \Gamma^*$ (i.e., the top of the stack is at the end) and a configuration~$c = (q, \gamma)$ of $\aut$ consists of a state~$q \in Q$ and a stack content~$\gamma$.  
The stack height of $c$ is $\sh(c) = \size{\gamma}-1$.
The initial configuration is $(q_\initmark, \bot)$.

A transition~$\tau = (q, X, a, q', \gamma' ) \in \Delta $ is enabled in a configuration~$c$ if $c = (q, \gamma X)$ for some $\gamma \in \Gammabot^*$.
In this case, we write~$(q, \gamma X) \trans{\tau} (q', \gamma\gamma')$.
A run of $\aut$ is a finite or infinite sequence~$r = c_0 \tau_0 c_1 \tau_1 c_2 \tau_2 \cdots$ of configurations and transitions with $c_n \trans{\tau_n}c_{n+1}$ for every $n$. A finite run is required to end with a configuration. 
A run is initial if it starts in the initial configuration.
An infinite run~$r= c_0 \tau_0 c_1 \tau_1 c_2 \tau_2 \cdots$ is a run of $\aut$ on $w \in \Sigma^\omega$, if $w = \ell(\tau_0 \tau_1 \tau_2 \cdots) $ (this implies that $\rho$ contains infinitely many $\Sigma$-transitions). 
We say that $r$ is accepting if there are infinitely $n$ such that the state of $c_n$ is in $F$, i.e., we consider Büchi acceptance.
The language~$L(\aut)$ recognized by an $\omega$-PDA~$\aut$ contains all $w \in \Sigma^\omega$ such that $\aut$ has an accepting run on $w$.

\emph{Visibly pushdown automata} are defined with respect to a partition~$(\Sigmacall, \Sigmareturn, \Sigmaskip)$ of the input alphabet into calls~(letters in $\Sigmacall$), returns~(letters in $\Sigmareturn$), and skips~(letters in $\Sigmaskip$) and have to satisfy the following conditions:
\begin{itemize}
    
    \item A letter~$a \in \Sigmacall$ is only processed by transitions of the form~$(q, X, a, q', XY)$ with $X\in \Gammabot$, i.e., some stack symbol~$Y$ is pushed onto the stack.
    
    \item A letter~$a \in \Sigmareturn$ is only processed by transitions of the form~$(q, X, a, q', \epsilon)$ with $X \neq \bot$ or the form~$(q, \bot, a, q',\bot)$, i.e., the topmost stack symbol is removed, or if the stack is empty, it is left unchanged.
    
    \item A letter~$a \in \Sigmaskip$ is only processed by transitions of the form~$(q, X, a, q',X)$ with $X \in \Gammabot$, i.e., the stack is left unchanged.

    \item There are no $\epsilon$-transitions.
    
\end{itemize}
Note that we allow, w.l.o.g., the automata to access the top stack symbol during calls and skips (see \cite[Section~2.1]{DBLP:conf/stoc/AlurM04}).
An $\omega$-PDA is a \emph{visibly} $\omega$-PDA ($\omega$-VPA for short), if there is a partition of its input alphabet into $(\Sigmacall, \Sigmareturn, \Sigmaskip)$ satisfying the conditions above.

Let $\Sigma$ be partitioned into $(\Sigmacall, \Sigmareturn, \Sigmaskip)$, let $\Sigma'$ be partitioned into $(\Sigmacall',\Sigmareturn', \Sigmaskip')$, and let $f\colon \Sigma \rightarrow \Sigma'$. We say that $f$ is a renaming, if $f(\Sigmacall) \subseteq \Sigmacall'$, $f(\Sigmareturn) \subseteq \Sigmareturn'$, and $f(\Sigmaskip) \subseteq \Sigmaskip'$, i.e., $f$ preserves the type of the letter.

\begin{proposition}[\cite{DBLP:conf/stoc/AlurM04}]
\label{prop_vpaclosure}
Let $\aut_0$ and $\aut_1$ be $\omega$-VPA over the same alphabet~$\Sigma$ with the same partition~$(\Sigmacall, \Sigmareturn, \Sigmaskip)$. Furthermore, let $f\colon \Sigma \rightarrow \Sigma'$ be a renaming where $\Sigma'$ is partitioned into $(\Sigmacall', \Sigmareturn', \Sigmaskip')$.
Then:
\begin{itemize}
    \item $L(\aut_0) \cap L(\aut_1)$ is recognized by an $\omega$-VPA with the partition~$(\Sigmacall, \Sigmareturn, \Sigmaskip)$ and polynomial size in $\size{\aut_0}$ and $\size{\aut_1}$.
    \item $L(\aut_0) \cup L(\aut_1)$ is recognized by an $\omega$-VPA with the partition~$(\Sigmacall, \Sigmareturn, \Sigmaskip)$ and polynomial size in $\size{\aut_0}$ and $\size{\aut_1}$.
    \item $\Sigma^\omega \setminus L(\aut_0)$ is recognized by an $\omega$-VPA with the partition~$(\Sigmacall, \Sigmareturn, \Sigmaskip)$ and exponential size in $\size{\aut_0}$.
    \item $f(L(\aut_0))$ is recognized by an $\omega$-VPA with the partition~$(\Sigmacall', \Sigmareturn', \Sigmaskip')$ and size~$\size{\aut_0}$.
\end{itemize}
\end{proposition}

\section{Logics for Context-Free Hyperproperties}
\label{sec_logic}

In this section, we introduce our logics and present some preliminary results.

\subsection{HyperPDA}
Let $\var = \set{\pi_0, \pi_1, \pi_2, \ldots}$ be \emph{the} set of trace variables.
A formula of \hylogicpda has the form~$\phi = \Q_0 \pi_0.\ Q_1 \pi_1.\ \ldots \Q_{k-1} \pi_{k-1}.\ \aut$ for some $k \ge 1$ where each $\Q_i$ is either an existential or universal quantifier and where $\aut$ is an $\omega$-PDA over an alphabet of the form~$\Sigma^k$. 
We call $\aut$ the automaton of $\phi$ and $k$ the arity of both $\phi$ and $\aut$.
As usual, we classify formulas of \hylogicpda by their quantifier alternations. Let $n \ge 1$.
$\Sigma_n$ ($\Pi_n$) contains all formulas with $n-1$ quantifier alternations beginning with an existential (universal) quantifier.

The semantics of \hylogicpda is defined with respect to a {trace assignment}, a partial mapping~$\Pi \colon \var \rightarrow \Sigma^\omega$. The assignment with empty domain is denoted by $\Pi_\emptyset$. Given a trace assignment~$\Pi$, a variable~$\pi$, and a trace~$t \in \Sigma^\omega$ we denote by $\Pi[\pi \rightarrow t]$ the assignment that coincides with $\Pi$ everywhere but at $\pi$, which is mapped to $t$. 

For sets~$T \subseteq \Sigma^\omega$ of traces and trace assignments~$\Pi$ we define 
\begin{itemize}
	\item $(T, \Pi) \models \exists \pi_j.\ \phi$ if there exists a trace~$t \in T$ such that $(T,\Pi[\pi_j \rightarrow t]) \models \phi$, 
	\item $(T, \Pi) \models \forall \pi_j.\ \phi$ if for all traces~$t \in T$: $(T,\Pi[\pi_j \rightarrow t]) \models \phi$, and
    \item $(T, \Pi) \models \aut$ if $\aut$ accepts $\combine{\Pi(\pi_0), \ldots, \Pi(\pi_{k-1})}$, where $k$ is the arity of $\aut$.
\end{itemize}

We say that $T$ {satisfies} a formula~$\phi$ if $(T, \Pi_\emptyset) \models \phi$. In this case, we write $T \models \phi$ and say that $T$ is a {model} of $\phi$. 
A transition system~$\tsys$ satisfies $\phi$, written $\tsys \models \phi$, if $\traces{\tsys}\models \phi$. 

\begin{remark}
\label{remark_hyltlsubsump}
\hylogicpda subsumes \hyltl, as quantifier-free \hyltl formulas can be translated into Büchi automata (which are $\omega$-PDA that do not use their stack), as they are (essentially) \ltl formulas. Hence, all lower bounds for \hyltl apply to \hylogicpda, e.g., \hylogicpda satisfiability is $\Sigma_1^1$-hard~\cite{DBLP:journals/lmcs/FortinKTZ25} and we conjecture that the problem is  $\Sigma_1^1$-complete.
\end{remark}

The \hylogicpda model-checking problem asks, given a transition system~$\tsys$ and a \hylogicpda formula~$\phi$, whether $\tsys \models \phi$. The model-checking problem for fragments~$\Sigma_n$ or $\Pi_n$ is defined by restricting the input formulas to the fragment.
The following result follows from emptiness of $\omega$-PDA being decidable respectively universality being undecidable, where the lower bound for $\Sigma_1$ is obtained by a reduction from the intersection problem for DFA~\cite{Kozen}.

\begin{theorem}
\label{thm_altfreefragments}
\hfill
\begin{enumerate}
    \item\label{item_existentialfragment} \hylogicpda model-checking for $\Sigma_1$ formulas is in \expt and \pspace-hard.
    \item\label{item_universalfragment} \hylogicpda model-checking for $\Pi_1$ formulas is undecidable.
\end{enumerate}
\end{theorem}

\begin{proof}
\ref{item_existentialfragment}.) We have $\tsys \models \exists \pi_0.\ \ldots \exists \pi_{k-1}.\ \aut$ if and only if 
\[
\set{\combine{t_0, \ldots, t_{k-1}} \mid t_0, \ldots, t_{k-1} \in \traces{\tsys}} \cap L(\aut) \neq \emptyset.
\]
As the set on the left-hand side of the intersection is $\omega$-regular, i.e., recognized by some Büchi automaton (see, e.g.,~\cite{GTW02} for definitions), and languages of $\omega$-PDA are effectively closed under intersections with $\omega$-regular languages~\cite{cg2}, the model-checking problem for $\Sigma_1$ formulas boils down to emptiness-checking for $\omega$-PDA.
The resulting $\omega$-VPA is of size~$\size{\tsys}^k \cdot\size{\aut}$. 
Thus, as emptiness can be decided in polynomial-time~\cite{pdaemptiness}, we obtain membership in \expt.

To prove the \pspace lower bound, we present a reduction from the intersection problem for DFA: given a sequence~$\autd_0, \ldots, \autd_{k-1}$ of DFA, determine whether $\bigcap_{i=0}^{k-1} L(\autd_i)$ is nonempty.
Kozen showed that this problem is \pspace-complete~\cite{Kozen}.

Given such a sequence~$\autd_0, \ldots, \autd_{k-1}$ of DFA (w.l.o.g., over some joint alphabet~$\Sigma$), one can construct a transition system~$\tsys$ such that 
\[
\traces{\tsys} = \bigcup_{i=0}^{k-1} \set{(w(0),i)\cdots (w(n-1),i) \#^\omega\cdots \mid w(0)\cdots w(n-1) \in L(\autd_i)} \cup (\Sigma\times\set{i})^\omega.
\]
To this end, one takes the disjoint union of the $\aut_i$, moves the transition labels of the DFA to the states (which requires to extend the state set), and adds $i$ to the label of the states resulting from $\aut_i$.

Furthermore, one can construct a Büchi automaton for the language
\[
\set{
\combine{w_0, \ldots, w_{k-1}} \mid \text{there exists } w(0) \cdots w(n-1) \in \Sigma^* \text{ s.t.\ } w_i = (w(0),i) \cdots (w(n-1),i) \#^\omega \text{ for all }i
}.
\]
This Büchi automaton can be turned into an equivalent $\omega$-VPA~$\aut$. Both $\tsys$ and $\aut$ are polynomial in the sum of the sizes of the $\autd_i$. For $\aut$, we rely on the fact that we allow automata to be incomplete, i.e., not all (exponentially many) letters must have transitions, only the (polynomially many) required to accept the words in $L(\aut)$ as described above.

Now, the languages of the $\aut_i$ have a nonempty intersection if and only if $\tsys \models \exists \pi_0.\ \ldots \exists \pi_{k-1}.\ \aut$.

\ref{item_universalfragment}.) 
Let $\aut$ be an $\omega$-PDA over some alphabet~$\Sigma$ and let $\tsys_U$ be a transition system with $\traces{\tsys_U} = \Sigma^\omega$, which can be constructed with $\size{\Sigma}$ many vertices.
Then, we have $\tsys_U \models \forall \pi_0.\ \aut$ if and only if $\aut$ is universal.
As universality for $\omega$-PDA is undecidable~\cite{cgdet}, so is the model-checking problem for $\Pi_1$ formulas.
\end{proof}

\begin{remark}
Note that we do not use the stack in the lower bound for \cref{thm_altfreefragments}.\ref{item_existentialfragment}, i.e., one could conjecture that the problem can be shown \expt-hard by utilizing the stack.
\end{remark}

As our undecidability result holds even for formulas with a single universal quantifier, we obtain undecidability for any quantifier-fragment that allows universal quantifiers.

\begin{corollary}
\hylogicpda model-checking is undecidable for all $\Sigma_n$ with $n>1$ and all $\Pi_n$ with $n \ge 1$, i.e., in particular for the full logic.
\end{corollary}

Our preliminary results show that a single universal quantifier makes model-checking undecidable, as  universality for $\omega$-PDA is undecidable. 
Thus, it is prudent to study formulas with restricted classes of $\omega$-PDA for which universality is decidable.
Thus, we restrict ourselves to $\omega$-VPA, which are closed under all Boolean operations, which implies in particular that universality is decidable.

\subsection{HyperVPA}
\hylogicvpa is the restriction of \hylogicpda to formulas whose automaton is an $\omega$-VPA.
For \hylogicvpa, we can refine the definition of the quantifier fragments~$\Sigma_n$ and $\Pi_n$.
To this end, let $\aut$ be an $\omega$-VPA of arity~$k > 0$ and let $I \subseteq \set{0,1,\ldots, k-1}$.
We say that $\aut$ is controlled by the indexes in $I$, if for all letters~$(a_0, a_1, \ldots, a_{k-1})$ and $(b_0, b_1, \ldots, b_{k-1})$ of $\aut$, $a_i = b_i$ for all $i \in I$ implies that they are both calls, or both returns, or both skips. 
Intuitively, the type of a letter of $\aut$ only depends on the positions in $I$.
The fragments~$\Sigma_n$ and $\Pi_n$ of \hylogicvpa are defined as for \hylogicpda.
Additionally, let $1 \le j \le n$.
$\Sigma_{n,j}$ and $\Pi_{n,j}$ contain the formulas of \hylogicvpa in $\Sigma_n$ and $\Pi_n$, respectively, whose automaton is controlled by the set of indexes of the $j$-th quantifier block.

While \hylogicvpa does not have a negation operator it is nevertheless closed under negation, since a negation can be pushed over quantifiers and since $\omega$-VPA are closed under complementation~\cite{DBLP:conf/stoc/AlurM04}.
Formally, given a formula~$\phi = \Q_0 \pi_0.\ Q_1 \pi_1.\ \ldots \Q_{k-1} \pi_{k-1}.\ \aut$ of \hylogicvpa, we define its negation~$\neg \phi$ as the formula~$\overline{\Q_0} \pi_0.\ \overline{Q_1} \pi_1.\ \ldots \overline{\Q_{k-1}} \pi_{k-1}.\ \overline{\aut}$ where $\overline{\exists} = \forall$, $\overline{\forall} = \exists$, and where $\overline{\aut}$ denotes an $\omega$-VPA accepting the complement of $L(\aut)$.
Note that $\overline{\aut}$ is defined with respect to the same partition of the input alphabet as $\aut$.
Hence, $\aut$ is controlled by some $I\subseteq \set{0,1,\ldots, k-1}$ if and only if $\overline{\aut}$ is controlled by $I$.
Note though that $\overline{\aut}$ may be exponentially larger than $\aut$~\cite{DBLP:conf/stoc/AlurM04}.

\begin{remark}
\label{remark_negprops}
Let $\phi$ be an \hylogicvpa formula, $n>0$, and $1 \le j \le n$.
\begin{enumerate}
    \item\label{remark_negprops_idem} $\phi $ and $\neg\neg\phi$ are equivalent.
    \item\label{remark_negprops_frag_n} $\phi$ is in $\Sigma_n$ if and only if $\neg \phi$ is in $\Pi_n$; and $\phi$ is in $\Sigma_{n,j}$ if and only if $\neg \phi$ is in $\Pi_{n,j}$.
    \item\label{remark_negprops_mc} $\tsys \models \phi$ if and only if $\tsys\not\models\neg\phi$.
\end{enumerate}
\end{remark}

\begin{remark}
\hylogicvpa subsumes \hyltl for the same reason \hylogicpda subsumes \hyltl (see \cref{remark_hyltlsubsump}).
Hence, \hylogicvpa satisfiability is $\Sigma_1^1$-hard and we again conjecture completeness.
\end{remark}

The decidability result for the $\Sigma_1$-fragment of \hylogicpda (\cref{thm_altfreefragments}.\ref{item_existentialfragment}) carries over to \hylogicvpa, as it is a fragment of \hylogicpda.
Similarly, the lower bound carries over, as it does not use the stack of the $\omega$-PDA. 

\begin{corollary}
\hylogicvpa model-checking for $\Sigma_1$ formulas is in \expt and \pspace-hard.
\end{corollary}

On the other hand, due to closure of $\omega$-VPA under complement, model-checking the $\Pi_1$- fragment of \hylogicvpa is also decidable.
Here, the lower bound follows from universality for VPA being \expt-complete~\cite{DBLP:conf/stoc/AlurM04}

\begin{theorem}
\label{thm_vpapione}
\hylogicvpa model-checking for $\Pi_1$ formulas is \expt-complete.
\end{theorem}

\begin{proof}
Let $\aut$ be an $\omega$-VPA.
We have $\tsys \models \forall \pi_0.\ \ldots \forall \pi_{k-1}.\ \aut$ if and only if 
\[
\set{\combine{t_0, \ldots, t_{k-1}} \mid t_0, \ldots, t_{k-1} \in \traces{\tsys}} \subseteq L(\aut),
\]
which is equivalent to 
\[
\set{\combine{t_0, \ldots, t_{k-1}} \mid t_0, \ldots, t_{k-1} \in \traces{\tsys}} \cap \overline{L(\aut)} = \emptyset,
\]
where $\overline{L(\aut)}$ is the complement of $L(\aut)$.

The \expt upper bound follows then from the following bounds:
\begin{itemize}
    \item $\omega$-VPA can be complemented with an exponential blow-up: Löding et al.~\cite{lms} showed that every $\omega$-VPA can be determinized into a stair VPA with an exponential blow-up, that stair VPA are complementable without a size blowup, and that stair VPA can be turned into equivalent standard $\omega$-VPA with a polynomial blowup. 

    Hence, one can construct an exponentially-sized $\omega$-VPA for $\overline{L(\aut)}$.
    
    \item One can construct, by taking $k$ copies of $\tsys$, a Büchi automaton for $\set{\combine{t_0, \ldots, t_{k-1}} \mid t_0, \ldots, t_{k-1} \in \traces{\tsys}}$, which can be turned into an equivalent $\omega$-VPA of exponential size.

    \item As $\omega$-VPA are closed under intersection, we can also construct an exponentially-sized $\omega$-VPA for $ \set{\combine{t_0, \ldots, t_{k-1}} \mid t_0, \ldots, t_{k-1} \in \traces{\tsys}} \cap \overline{L(\aut)}$.

    \item Emptiness of $\omega$-PDA (and thus of $\omega$-VPA) can be checked in polynomial time~\cite{pdaemptiness}.
\end{itemize}

For the matching lower bound, recall that universality of $\omega$-VPA is \expt-complete~\cite{DBLP:conf/stoc/AlurM04}. Thus, noting that an $\omega$-VPA $\aut$ (with alphabet~$\Sigma$) is universal if and only if $\tsys_\Sigma \models \forall \pi_0.\ \aut$, where $\tsys_\Sigma$ is a transition system with $\traces{\tsys_\Sigma} = \Sigma^\omega$, yields the desired reduction from $\omega$-VPA universality to \hylogicvpa model-checking for $\Pi_1$ formulas.
\end{proof}

Thus, we have shown that model-checking is decidable for all alternation-free fragments of \hylogicvpa.
In the next two subsections, we consider formulas with alternations. 
In \cref{subsec_dec}, we show that model-checking is decidable for the fragments~$\Sigma_{n,1}$ and $\Pi_{n,1}$ for all $n > 1$, i.e., if the automaton is controlled by the first quantifier block.
Dually, in \cref{subsec_undec}, we show that model-checking is undecidable for the fragments~$\Sigma_{2,2}$ and $\Pi_{2,2}$, i.e., if the automaton is controlled by the second quantifier block.
Finally, using these results, we show at the end of \cref{subsec_undec} that model-checking is undecidable for the fragments~$\Sigma_{n,j}$ and $\Pi_{n,j}$ for all $n > 1$ and $j>1$.
Thus, we present a complete picture of the decidability landscape for \hylogicvpa.

\subsection{Fragments of HyperVPA with Decidable Model-Checking}
\label{subsec_dec}

Our main result of this section is the decidability of \hylogicvpa model-checking for formulas where the stack height depends only on the traces in the first quantifier block, but with arbitrary quantifier alternation. 
To this end, we generalize the automata-based approach to \hyltl model-checking~ \mbox{\cite{ClarksonFKMRS14,FinkbeinerRS15}} from $\omega$-regular automata to $\omega$-VPA. 
This approach relies on quantifier-free formulas being translatable into automata which are closed under intersection with finite transition systems and projection (to handle existential quantification), and closed under complemenation (to handle universal quantification by a reduction to existential quantification). While $\omega$-VPA are closed under intersection and complementation, they are in general not closed under projections, as these can destroy the partition into calls, returns, and skips. 
However, $\omega$-VPA are closed under so-called renamings, projections that preserve the partition. 
Such renamings are sufficient to generalize the automata-based construction to \hylogicvpa  formulas where the stack height depends only on the traces in the first quantifier block, as the projections are applied \myquote{from the other end}.

In \cref{subsec_undec}, we show that this fragment is the largest one (in terms of which quantifiers the stack height may depend on) with decidable model-checking.

\begin{theorem}
Let $n>1$. \hylogicvpa model-checking for $\Sigma_{n,1}$ and $\Pi_{n,1}$ formulas is decidable in $n$-fold exponential time.
\end{theorem}

\begin{proof}
We restrict ourselves to the fragment~$\Sigma_{n,1}$, as the fragment~$\Pi_{n,1}$ is then easily handled by dualization (see  \cref{remark_negprops_frag_n} of \cref{remark_negprops}), as we have $\tsys \models \varphi$ if and only if $\tsys \not \models \neg \varphi$ (see \cref{remark_negprops_mc} of \cref{remark_negprops}).

So, fix a transition system~$\tsys = (V,E,V_\initmark, \lambda)$ with $\lambda\colon V \rightarrow \Sigma$, and a $\Sigma_{n,1}$ sentence 
\[\phi = Q_0\pi_0.\ Q_1\pi_1 \ldots Q_{k-1}\pi_{k-1}.\ \aut\]
of \hylogicvpa for some $n>1$ such that $\aut$ processes words over $\Sigma^k$.

Let $e_0$ be maximal with $Q_i = \exists$ for all $0 \le i < e_0$, i.e., $e_0$ is the number of initial existential quantifiers. 
As $\varphi$ is a $\Sigma_{n,1}$ sentence with $n>1$, we have $0 < e_0 < k$.
By definition, for all $(a_0, a_1, \ldots, a_{k-1}) \in \Sigma^k$ and $(b_0, b_1, \ldots, b_{k-1}) \in \Sigma^k$, $a_i = b_i$ for all $i < e_0$ implies that they are both calls, or both returns, or both skips. 
For $e_0 \le i < k$, we define the partition of $\Sigma^i$ into calls, returns, and skips induced by projecting the letters of $\Sigma^k$ by removing the last $k-i$ components, which is well-defined.

For $e_0 \le i < k$, we define
\[
\phi_i = Q_i\pi_i.\ Q_{i+1}\pi_{i+1} \ldots Q_{k-1}\pi_{k-1}.\ \aut,
\]
as well as $\phi_k = \aut$.
Note that the free variables of each $\phi_i$ with $e_0 \le i < k$ are exactly $\pi_0, \ldots, \pi_{i-1}$, which for every such $i$ includes the variables~$\pi_0,\ldots,\pi_{e_0-1}$ of the initial block of existential quantifiers (and for $i = e_0$, the free variables are exactly those of the initial block of existential quantifiers).
Each such $\phi_i$ induces the language
\[\lang{\phi}{\tsys}{i} =  \{\combine{\Pi(\pi_0), \ldots, \Pi(\pi_{i-1})} \mid \Pi(\pi_j) \in \Sigma^\omega \text{ for $0 \le j < i$ and }  (\traces{\tsys},\Pi) \models \phi_i\}
\]
over the alphabet~$\Sigma^i$.
Note that the $\Pi(\pi_j)$ with $0 \le j < i$ are in $\Sigma^\omega$ while the $\pi_j$ with $i \le j < k$ range, by definition of the \hylogicvpa semantics, over $\traces{\tsys}$.

We (effectively) construct, by induction from $i = k$ to $i=e_0$, an $\omega$-VPA~$\aut_i$ recognizing $\lang{\phi}{\tsys}{i}$ (using the partition of $\Sigma^i$ into calls, returns, and skips introduced above).

The induction start $i = k$ is trivial:
As $\phi_k = \aut$ does not have any quantifiers, $(\traces{\tsys},\Pi) \models \phi_k$ only depends on $\Pi$.
In fact, we have (by the semantics of \hylogicvpa) that $(\traces{\tsys},\Pi) \models \phi_k$ if and only if $\aut$ accepts $\combine{\Pi(\pi_0), \ldots, \Pi(\pi_{k-1})}$.
Thus,
\[
\lang{\phi}{\tsys}{k} =  \{\combine{\Pi(\pi_0), \ldots, \Pi(\pi_{k-1})} \mid \Pi(\pi_j) \in \Sigma^\omega \text{ for $0 \le j < k$ and }  (\traces{\tsys},\Pi) \models \phi_k\}
\]
is equal to $L(\aut)$.
Hence, we define $\aut_k = \aut$.

For the induction step with $e_0 \le i<k$, we first consider the case where $Q_i$, the outermost quantifier of $\phi_i$, is existential, i.e., $\phi_i = \exists \pi_i.\ \phi_{i+1}$.
Let $\aut_{i+1} = (Q, \Sigma^{i+1},\Gamma, q_\initmark, \Delta, F)$ be the $\omega$-VPA obtained by induction hypothesis, i.e., such that 
\[
L(\aut_{i+1}) = \set{ \combine{\Pi(\pi_0), \ldots, \Pi(\pi_{i})} \mid \Pi(\pi_j) \in \Sigma^\omega \text{ for all $0 \le j < i+1$ and }   (\traces{\tsys},\Pi) \models \phi_{i+1} }
.
\]
It uses the partition induced by projecting $\Sigma^k$ to $\Sigma^{i+1}$.

We need to construct an $\omega$-VPA~$\aut_i$ that accepts~$\combine{\Pi(\pi_0), \ldots, \Pi(\pi_{i-1})}$ if and only if $\aut_{i+1}$ accepts $\combine{\Pi(\pi_0), \ldots, \Pi(\pi_{i}), t}$ for some $t \in \traces{\tsys}$.
We do so by taking the product of $\aut_{i+1}$ and $\tsys$ where we \myquote{synchronize} the last component of $\aut_{i+1}$ and a trace of an initial path through $\tsys$. This component is then projected away.

Formally, we define $\aut_{i+1} \oplus \tsys = (Q \times V \cup \set{q_\initmark'} ,  \Sigma^i,q_\initmark',\Gamma, \Delta', F \times V)$ where $q_\initmark'$ is a fresh initial state that is not in $Q \times V$ and where $\Delta'$ contains two types of transitions:
\begin{itemize}
    \item Initial transitions: \[
(q_\initmark', X, \vectordots{a_0}{a_{i-1}}, (q',v'),\gamma)
\]
is in $\Delta'$ if and only if 
\[
(q_\initmark, X, \vectordotswithtwoargsatend{a_0}{a_{i-1}}{a}, q',\gamma)
\]
is in $\Delta$, $(v,v') \in E$ for some $v \in V_\initmark$, and $\lambda(v) = a$. 

    \item Non-initial transitions: \[
((q,v), X, \vectordots{a_0}{a_{i-1}}, (q',v'),\gamma)
\]
is in $\Delta'$ if and only if 
\[
(q, X, \vectordotswithtwoargsatend{a_0}{a_{i-1}}{a}, q',\gamma)
\]
is in $\Delta$, $(v,v') \in E$, and $\lambda(v) = a$.

\end{itemize}
Note that $\aut_{i+1} \oplus \tsys$ uses the partition induced by projecting $\Sigma^k$ to $\Sigma^i$.   

Then, $\aut_{i+1} \oplus \tsys $ indeed recognizes 
\[\set{ \combine{\Pi(\pi_0), \ldots, \Pi(\pi_{i-1})} \mid \Pi(\pi_j) \in \Sigma^\omega \text{ for all $0 \le j < i$ and }    (\traces{\tsys},\Pi) \models \phi_i }.\]
Hence, we can define $\aut_i =\aut_{i+1} \oplus \tsys $.

Finally, consider the induction step with $e_0 \le i<k$ where $Q_i$, the outermost quantifier of $\phi_i$, is universal, i.e., $\phi_i = \forall \pi_i.\ \phi_{i+1}$.
Let $\aut_{i+1}$ be the automaton obtained by induction hypothesis. Then, $\compl{\compl{\aut_{i+1}}\oplus\tsys} $ recognizes \[\set{ \combine{\Pi(\pi_0), \ldots, \Pi(\pi_{i-1})} \mid   \Pi(\pi_j) \in \Sigma^\omega \text{ for all $0 \le j < i$ and }   (\traces{\tsys},\Pi) \models \phi_i },\] where $\compl{\aut}$ denotes a complement automaton of $\aut$.
Hence, we define $\aut_i = \compl{\compl{\aut_{i+1}}\oplus\tsys}$.
This completes the induction.

Recall that the first $e_0 >0$ quantifiers in $\varphi$ are existential.
Hence, when evaluating $\phi$ on $\tsys$, they range over $\traces{\tsys}$.
Thus, we have $\tsys \models \phi$ if and only if 
\[
\set{\combine{t_0, \ldots, t_{e_{0}-1}} \mid t_0, \ldots, t_{e_0-1} \in \traces{\tsys} } \cap \lang{\phi}{\tsys}{e_0}
\]
is nonempty.
The latter condition can be checked by constructing $\aut_{e_0}$ (as described above), an $\omega$-VPA recognizing $\set{\combine{t_0, \ldots, t_{e_{0}-1}} \mid t_0, \ldots, t_{e_0-1} \in \traces{\tsys} }$ (see the proof of \cref{thm_vpapione}), and then checking their intersection automaton (see \cref{prop_vpaclosure}) for emptiness~\cite{pdaemptiness}.

It remains to analyze the complexity of the algorithm. 
The dominant factor in the construction of $\aut_{e_0}$ is the number of complementations.
Note that a maximal block of universal quantifiers only needs two complementations, the other ones cancel each other out.
Hence, the number of complementations can be bounded by the number of alternations, which yields the desired upper bound using the closure properties listed in \cref{prop_vpaclosure}.
\end{proof}


\subsection{Fragments of HyperVPA with Undecidable Model-Checking}
\label{subsec_undec}

In this subsection, we consider the model-checking problem for the fragment~$\Pi_{2,2}$, i.e., formulas of the form~$\forall^*\exists^*.\ \aut$ such that the stack height in $\aut$ only depends on the existentially quantified traces.
Here, one can encode the undecidable universality problem for pushdown automata by letting the universal quantifier range over inputs~$w$ and the existential one over runs~$r$, i.e., $\aut$ checks that $r$ is indeed an accepting run on $w$ by simulating $r$. Hence, the stack of $\aut$ is indeed controlled by the existentially quantified variable.

At the end of the subsection, we generalize this result to arbitrary formulas where the stack height is controlled by a non-initial quantifier block, showing that the decidable fragment presented in \cref{subsec_dec} is (in terms of which quantifiers the stack height may depend on) maximal.

\begin{theorem}
\label{thm_undec}
\hylogicvpa model-checking for $\Pi_{2,2}$ formulas is undecidable.
\end{theorem}

\begin{proof}
Let $\aut$ be an $\omega$-PDA over $\Sigma$ with set~$\Delta$ of transitions and consider the language~$L$ containing all words of the form~$\combine{w, \tau_0 \tau_1 \tau_2\cdots}$ with $w \in \Sigma^\omega$ and $\tau_0 \tau_1 \tau_2\cdots \in \Delta^\omega$ such that $\tau_0 \tau_1 \tau_2 \cdots$ induces an accepting run of $\aut$ on $w$. 
Assuming that $\aut$ does not have $\epsilon$-transitions implies that $\tau_j$ processes $w(j)$, i.e., the two components of $\combine{w, \tau_0 \tau_1 \tau_2\cdots}$ are synchronized. 
Intuitively, we construct an $\omega$-VPA~$\autb$ recognizing $L$ that simulates the run induced by the second component of the input (if it does indeed induce a run) and accepts if the run processes the word in the first component and is accepting (we are omitting some technical details for now).
Note that $\autb$ is indeed controlled by the second component, as $\autb$ simulates the transitions in the second component.
Then, a suitable transition system satisfies $\forall \pi_0.\ \exists \pi_1.\ \autb$ if and only if $L(\aut)$ is universal, as the formula expresses that for every input~$w \in \Sigma^\omega$ there is an accepting run of $\aut$ on $w$. 
As universality of pushdown automata is undecidable, this yields the desired result.

To ensure that $L$ as described above can indeed be recognized by an $\omega$-VPA, it is convenient to start with context-free grammars instead of pushdown automata.
But we will also work with PDA (over finite words) that accept with an empty stack. 
Such automata have the form~$\aut = (Q, \Sigma, \Gamma, q_\initmark, X_\initmark, \Delta)$ where $Q$, $\Sigma$, $\Gamma$, $q_\initmark$ are as for $\omega$-PDA as defined in \cref{sec_prelims} and where $X_\initmark \in \Gamma$ is the initial stack symbol and where $\Delta$ is a finite subset of $Q \times \Gamma \times \Sigma_\epsilon \times Q \times \Gamma^*$. The initial configuration is $(q_\initmark, X_\initmark)$ and a configuration is accepting if its stack is empty (i.e., we do not have a dedicated stack bottom symbol).
In the following, we disregard the empty word for technical reasons, i.e., we only consider nonempty words. 

We say that a PDA or an $\omega$-PDA~$\aut$ is normalized if it does not have any $\epsilon$-transitions and if every transition increases the stack height by at most one during each transition, i.e., every transition~$(q, X, a, q', \gamma)$ satisfies $\size{\gamma} \le 2$.
Furthermore, we say that $\aut$ is fully normalized if it is normalized and additionally 
\begin{itemize}
    \item every transition of $\aut$ of the form~$(q, X, a, q', X')$ satisfies $X = X'$, and
    \item every transition of $\aut$ of the form~$(q, X, a, q', X'X'')$ satisfies $X = X'$.
\end{itemize}
Note that these two conditions are the same as for visibly pushdown automata.
Our first goal is to show that universality is undecidable for fully normalized $\omega$-PDA.
This result relies on several textbook constructions that we need to carefully combine. 
For the sake of self-containedness, we present these (known) constructions here.

Universality for context-free grammars (generating languages of nonempty finite words) is undecidable. 
Every context-free grammar can be transformed into an equivalent one in 2-Greibach normal form~\cite[Corollary 3.2]{gnf}, i.e., every rule has the form~$A \rightarrow a A_1 \cdots A_n$ where $a$ is a terminal, $A, A_1, \cdots, A_n$ are nonterminals, and $n \le 2$. 
Hence, universality is also undecidable for grammars in 2-Greibach normal form (here, we benefit from disregarding the empty word, as we do not need to allow a special rule to generate the empty word).
Applying the classical translation (see, e.g.,~\cite{HMU}) of context-free grammars into PDA that accept with an empty stack to grammars in 2-Greibach normal form yields normalized PDA's (note that the empty word can only be accepted with $\epsilon$-transitions, as the stack is initially nonempty, but we accept with an empty stack, i.e., the PDA would not be normalized).
Hence, universality is also undecidable for normalized PDA's.

Next, we turn a PDA into an $\omega$-PDA while preserving universality and normalization.
Formally, given a PDA~$\aut = (Q, \Sigma, \Gamma, q_\initmark, X_\initmark, \Delta)$ we construct the $\omega$-PDA~$\aut' = (Q', \Sigma', \Gamma', q_\initmark', \Delta', F')$ where (see also \cref{fig_automatatransform})
\begin{itemize}
    \item $Q' = Q \cup \set{q_\initmark', q_{s_1}', q_{s_2}', q_t'}$,
    \item $\Sigma' = \Sigma \cup \set{\#}$ where $\# \notin \Sigma$ is a fresh input letter, 
    \item $\Gamma' = \Gamma$, and
    \item $\Delta'$ contains the following transitions:
        \begin{itemize}
            \item $(q_\initmark', \bot, a, q_\initmark', \bot)$ for all $a \in \Sigma$: A self-loop on the new initial state with every non-$\#$ letter.

            \item $(q_\initmark', \bot, \#, q_\initmark, \# X_\initmark)$: A $\#$-transition to the (old) initial state of $\aut$ that puts the initial stack symbol of $\aut$ on top of the stack bottom symbol~$\bot$ of $\aut'$.

            \item All transitions in $\Delta$: Every run of $\aut'$ can be simulated, and if (and only if) the stack of $\aut$ is emptied during the simulation, the stack bottom symbol of $\aut'$ is exposed again.

            \item $(q, \bot, \#, q_{s_1}', \bot)$ for every state~$q \in Q$ of $\aut$: Once the simulation has ended with an empty stack, the new (sink) state~$q_{s_1}'$ can be reached by processing another~$\#$.

            \item $(q_{s_1}', \bot, a, q_{s_1}', \bot)$ for all $a \in \Sigma \cup \set{\#}$: A self-loop on $q_{s_1}'$ with every letter in $\Sigma \cup \set{\#}$.

            \item $(q_\initmark', \bot, \#, q_{s_2}'. \bot)$: A $\#$-transition from the new initial state to another new (sink) state~$q_{s_2}'$.

            \item $(q_{s_2}', \bot, a, q_{s_2}', \bot)$ for all $a \in \Sigma$: A self-loop on $q_{s_2}'$ with every letter in $\Sigma$, but not $\#$.

            \item $(q_\initmark', \bot, \#, q_t', \bot)$ and $(q_t', \bot, \#, q_{s_1}', \bot)$: Two transitions processing $\#\#$ from the new initial state to the first sink state~$q_{s_1}'$. 
            
        \end{itemize}

    \item Finally, we define $F' = \set{q_\initmark',q_{s_1}', q_{s_2}'}$, i.e., a run is accepting if when it enters $\aut$ (which can happen only once), then it also leaves it again. If it is never entered, then the run is accepting as well.
\end{itemize}
The resulting~$\omega$-PDA~$\aut'$ is indeed normalized if $\aut$  is normalized.

\begin{figure}
    \centering

    \begin{tikzpicture}[thick]
  \draw[thick,rounded corners] (1,0) rectangle (6,2);
  
    \node[anchor=north west] at (1,2) {$\aut$};
    
  \node[state,accepting] (qs1) at (8.5,1) {$q_{s_1}'$};

  \node[state,accepting] (qiprime) at (-1.5,1) {$q_\initmark'$};
  \node[state,accepting] (qs2) at (-4,1) {$q_{s_2}'$};

  \node[state] (qi) at (2,1) {$q_\initmark$};

  \node[state] (qt) at (3.5,2.75) {$q_t'$};

  \path[->, > = stealth]
  (-1.5,0) edge (qiprime)
  (qiprime) edge[loop above] node[above] {$*, \bot \mid \bot $} ()
  (qiprime) edge node[above] {$\#,\bot \mid \bot$} (qs2)
  (qiprime) edge node[above,xshift=-6] {$\#,\bot \mid \bot X_\initmark$} (qi)
  (qs2) edge[loop above] node[above] {$*, \bot \mid \bot $} ()
  (5.5,1) edge[line width=2] node[above,xshift = 6] {$\#, \bot \mid \bot $} (qs1)
  (qiprime) edge[bend left=17] node[above,yshift=6] {$\#,\bot\mid\bot$} (qt)
  (qt) edge[bend left=17] node[above,yshift=6] {$\#,\bot\mid\bot$} (qs1)
  (qs1) edge[loop above] node[above] {$\dagger, \bot \mid \bot $} ()
  ;
    \end{tikzpicture}

    \caption{Turning an PDA into an $\omega$-PDA while preserving universality and normalization. A transition~$(q, X, a, q', \gamma)$ is depicted by an edge from $q$ to $q'$ labeled by $a, X \mid \gamma$.
    Here, $*$ represents an arbitrary letter from $\Sigma$ (but not $\#$!) and $\dagger$ represents an arbitrary letter from $\Sigma\cup\set{\#}$. Furthermore, the thick transition connects every state of $\aut$ to $q_{s_1}'$.}
    \label{fig_automatatransform}
\end{figure}

Now, consider an $\omega$-word~$w$ over $\Sigma\cup\set{\#}$.
If $w$ does not contain a $\#$ then it is accepted by $\aut'$ using the self-loops on the new initial state~$q_\initmark'$. 
Also, if $w$ contains exactly one $\#$ then is accepted by $\aut'$ using the self-loops on the new initial state until the $\#$ is processed using the transition leading to $q_{s_2}'$, from where the remainder of the word can be processed (which does not contain a $\#$).

So, let us consider the case where $w$ has at least two $\#$. 
If the first two $\#$ occur at consecutive positions then the word is accepted by using the self-loops on the new initial state~$q_\initmark'$ until the first $\#$ is processed using the transition to $q_t'$.
This is, by our assumption, immediately followed by a second $\#$, which is processed by the transition to $q_{s_1}'$.
From there, the remainder of the word can be processed.

Finally, assume the first two $\#$ do not occur at consecutive positions, i.e., $w$ has the form~$x_0 \# x_1 \# w' $ such that $x_0, x_1$ are finite words, $x_1$ is nonempty, and $w'$ is an $\omega$-word over $\Sigma\cup\set{\#}$.
The only way to accept such a word is to process $x_0$ using the self-loop on the new initial state~$q_\initmark'$ until the first $\#$ is processed using the transition to $q_\initmark$.
From there, an accepting run of $\aut$ on $x_1$ has to be simulated, as the only way to process the second~$\#$ is to take a $\#$-transition to $q_{s_1}'$, which is only enabled if the stack only contains~$\bot$ (i.e., in the simulation, the stack of $\aut$ has been emptied).
Hence, we have indeed $L(\aut) = \Sigma^+$ if and only if $L(\aut') = (\Sigma\cup \set{\#})^\omega$.
Hence, universality for normalized $\omega$-PDA is undecidable.

As a next step, we need to fully normalize a normalized $\omega$-PDA~$\aut'$.
This can be achieved by simulating~$\aut'$ while storing the topmost stack symbol of $\aut'$ in the state and the rest of the stack of $\aut'$ in the stack of the simulating automaton.

Formally, given a normalized $\omega$-PDA~$\aut' = (Q', \Sigma', \Gamma', q_\initmark', \Delta', F')$ we define the \newline$\omega$-PDA~$\aut'' = (Q'', \Sigma'', \Gamma'', q_\initmark'', \Delta'', F'')$ with
\begin{itemize}
    \item $Q'' = Q' \times \Gamma$,
    \item $\Sigma'' = \Sigma'$,
    \item $\Gamma'' = \Gamma'$,
    \item $q_\initmark'' = (q_\initmark, \bot)$, and
    \item $\Delta''$ contains the following transitions:
        \begin{itemize}
            \item For every transition~$(q, X, a, q', \epsilon) \in \Delta'$ of $\aut'$ (which implies~$X \neq \bot$), $\aut''$ has the transition~$((q,X),Y,a,(q', Y),\epsilon)$ for every $Y \in \Gamma$ and the transition~$((q, X),\bot,a,(q', \bot),\bot)$: The topmost symbol of the stack (encoded in the state of $\aut''$) is removed and updated.

            \item For every transition~$(q, X, a, q', X') \in \Delta'$ of $\aut'$, $\aut''$ has the transition~$(q, X), Y, a, (q', X'), Y)$ for all $Y \in \Gamma \cup \set{\bot}$: The topmost symbol of the stack (encoded in the state of $\aut''$) is updated.

            \item For every transition~$(q, X, a, q', X'X'') \in \Delta'$ of $\aut'$ with $X \neq \bot$, $\aut''$ has the transition $((q, X), Y, a, (q', X''), YX')$ for all $Y \in \Gamma\cup \set{\bot}$: The topmost symbol of the stack (encoded in the state of $\aut''$) is updated and then pushed on the stack (i.e., $X$ is replaced by $X'$, which is pushed) and the new topmost symbol of the stack is $X''$.

            \item For every transition~$(q, \bot, q', X'X'') \in \Delta'$ of $\aut'$ (which implies $X' = \bot$), $\aut''$ has the transition~$((q,\bot),\bot, a, (q', X''),\bot )$: If a symbol ($X''$ here) is pushed on the stack containing only the stack bottom symbol, then it is stored in the state.
        \end{itemize}
    \item Finally, $F'' = F' \times \Gamma$.
\end{itemize}
Note that $\aut''$ is fully normalized by construction. 

Given a configuration~$c = (q, \bot X_1 \cdots X_n)$ of $\aut'$ with $n>0$, we define the configuration~$f(c) = ((q, X_n), \bot X_1 \cdots X_{n-1})$.
Furthermore, for $c = (q, \bot)$, we define $f(c) = ((q, \bot), \bot)$.
Note that $f$ is a bijection between configurations of $\aut'$ and $\aut''$.
Now, an induction over $w \in \Sigma^*$ shows that if $c_0 c_1 \cdots c_n$ are the configurations of a run prefix of $\aut'$ processing $w$, then $f(c_0) f(c_1) \cdots f(c_n)$ are the configurations of a run prefix of $\aut''$ processing $w$.
Dually, another induction over $w \in \Sigma^*$ shows that if $c_0 c_1 \cdots c_n$ are the configurations of a run prefix of $\aut''$ processing $w$, then $f^{-1}(c_0) f^{-1}(c_1) \cdots f^{-1}(c_n)$ are the configurations of a run prefix of $\aut'$ processing $w$.
From these two translations and the definition of $F''$, we can conclude~$L(\aut'') = L(\aut')$.
Hence, universality is undecidable for fully normalized $\omega$-PDA.

Now, given a fully normalized $\omega$-PDA~$\aut'' = (Q, \Sigma, \Gamma, q_\initmark, \Delta, F)$ (we drop the primes for the sake of readability), we can construct an $\omega$-VPA~$\autb$ recognizing the $\omega$-language
\begin{align*}
L(\autb) = &\left\{
  \binom{w_0}{\tau_0}
  \binom{w_1}{\tau_1}
  \binom{w_2}{\tau_2} \cdots \in (\Sigma \times \Delta)^\omega
  \,\middle|\,
\begin{minipage}{6.5cm}
\text{$\tau_0 \tau_1 \tau_2 \cdots$ induces an initial accepting run}\\
  \text{of $\aut''$ processing $w_0 w_1 w_2 \cdots$}
\end{minipage}
\right\}\cup\\
&
\left\{
  \binom{w_0}{\tau_0}
  \binom{w_1}{\tau_1}
  \binom{w_2}{\tau_2} \cdots \in ((\Sigma \cup \Delta) \times \Delta)^\omega
  \,\middle|\,
  w_n \in \Delta \text{ for some } n
\right\}
\end{align*}
with respect to the following partition of $(\Sigma \cup \Delta) \times (\Sigma \cup \Delta)$, where $a$ is an arbitrary letter in $\Sigma \cup \Delta$.
\begin{itemize}
    \item $\binom{a}{\tau}$ is in $\Sigmacall$ if $\tau \in \Delta$ pushes a letter on the stack, i.e., it is of the form~$(q, X, a', q', XX')$.
    \item $\binom{a}{\tau}$ is in $\Sigmareturn$ if $\tau \in \Delta$ pops a letter off the stack, i.e., it is of the form~$(q, X, a', q', \epsilon)$.
    \item $\binom{a}{\tau}$ is in $\Sigmaskip$ if $\tau \in \Delta$ does not change the stack, i.e., it is of the form~$(q, X, a', q', X)$.
    \item $\binom{a}{b}$ with $b \in \Sigma$ is in $\Sigmaskip$.
\end{itemize}
Thus, $\autb$ is controlled by the second component of the input letters.

Intuitively, $\autb$ simulates the run induced by $\tau_0 \tau_1 \tau_2 \cdots$, checks that each $\tau_n$ processes $w_n \in \Sigma$ (here, we crucially rely on $\aut''$ being fully normalized), and gets stuck when $\tau_0 \tau_1 \tau_2 \cdots$ does not induce an initial run on $w_0w_1w_2\cdots$. It accepts if the simulated run is accepting or if $w_0 w_1 w_2 \cdots$ contains a letter from $\Delta$.

More formally, $\autb$ uses the same states as $\aut''$, as well as a fresh sink state called~$q_s$. The initial state is the same, as are the accepting ones ($q_s$ is accepting as well).
Also, $\autb$ uses the same stack alphabet as $\aut''$ and each transition~$\tau = (q, X, a, q', \gamma)$ is turned into the transition~$(q, X, (a,\tau) , q', \gamma)$ and we add all transitions of the form~$(q, X, (\tau, *), q_s, \gamma)$ with $\tau \in \Delta$ for suitable $\gamma$, and transitions of the form~$(q_s, X, (*,*), q_s, \gamma)$ for suitable $\gamma$, where $*$ stands for an arbitrary letter. These transitions allow $\autb$ to accept if the first component contains a letter from $\Delta$.

Now, we claim that $(\Sigma \cup \Delta)^\omega \models \forall \pi_0.\ \exists \pi_1.\ \autb$ if and only if $\aut''$ is universal.
This concludes the proof, as $ \forall \pi_0.\ \exists \pi_1.\ \autb$ is in $\Pi_{2,2}$ and a finite transition system with language~$(\Sigma \cup \Delta)^\omega $ can trivially be constructed.

So, let $\aut''$ be universal. Then, no matter how a trace~$t_0$ for $\pi_0$ is selected, we can find a trace~$t_1$ for $\pi_1$ so that $\autb$ accepts the pair~$\combine{t_0, t_1}$:
\begin{itemize}
    \item If $t_0$ contains a letter from $\Delta$, then $\combine{t_0, t_1}$ is accepted by $\autb$, independently of the choice of $t_1$.
    \item If $t_0$ does not contain a letter from $\Delta$, i.e., only letters from $\Sigma$, then it is accepted by $\aut''$. Hence, there is an infinite sequence~$t_1$ of transitions of $\aut''$ inducing an initial accepting run of $\aut''$ processing $t_0$. Hence, $\autb$ accepts $\combine{t_0, t_1}$.
\end{itemize}

On the other hand, if $\aut''$ is not universal, then there is a trace~$t_0 \notin L(\aut'')$. 
Thus, for every $t_1 \in (\Sigma \cup \Delta)^\omega$, $\autb$ does not accept $\combine{t_0, t_1}$, i.e., 
$(\Sigma \cup \Delta)^\omega \not\models \forall \pi_0.\ \exists \pi_1.\ \autb$.    
\end{proof}

Finally, note that our construction in the proof of \cref{thm_undec} uses $\omega$-VPA with Büchi acceptance.
However, one can strengthen our result to use $\omega$-VPA with weak parity acceptance~\cite[Chapter 1]{GTW02} where states are labeled by natural numbers (so-called colors) and a run is accepting if the maximal color occurring during the run is even:
\begin{itemize}
    \item In $\aut'$, we color $q_\initmark'$ and $q_{s_2}'$ both with color~$0$, the states of $\aut$ (and $q_t'$) with color~$1$ (so they have to be left eventually), and the state~$q_{s_1}'$ by color~$2$ (so that reaching it implies the run is accepting).
    \item In the construction of $\aut''$, we do not need to introduce new colors, even when using weak parity acceptance, as $\aut''$ just simulates $\aut'$ step-by-step. Formally, the color of $(q,X)$ in $\aut''$ is the color of $q$ in $\aut'$.
    \item The construction of $\autb$ uses the colors of $\aut''$ (when simulating $\aut''$) and the colors $2$ for the sink state~$q_s$ that is reached if the first component contains a letter from $\Delta$.
\end{itemize}
Hence, model-checking is undecidable even for $\Sigma_{2,2}$ formulas and $\Pi_{2,2}$ formulas whose automaton is an $\omega$-VPA with weak parity acceptance with three colors.
We leave it open whether this result can be further improved to show undecidability for two colors, i.e., for safety or reachability acceptance.

Again, by applying \cref{remark_negprops} we also obtain undecidability for the dual fragment~$\Sigma_{2,2}$.
Further, as $\Sigma_{2,2}$ or $\Pi_{2,2}$ can be embedded in each $\Sigma_{n,j}$ and each $\Pi_{n,j}$ with $n>2$ and $j>1$ (using dummy variables), model-checking for these fragments is undecidable as well.

\begin{corollary}
Let $n > 1$ and $1 < j \le n$. Then,
\hylogicvpa model-checking for $\Sigma_{n,j}$ and $\Pi_{n,j}$ formulas is undecidable.
\end{corollary}

\section{Explanations for \boldmath$\Pi_{2,1}$ and Counterexamples for \boldmath$\Sigma_{2,1}$}
\label{sec_prophecies}

Consider a \hylogicvpa formula~$\varphi = \forall \pi_0.\ \exists \pi_1.\ \aut$ and a transition system~$\tsys$ such that $\tsys \models \varphi$.
Then, there exists a Skolem function~$f \colon \traces{\tsys} \rightarrow \traces{\tsys}$ witnessing this fact, i.e., such that $[\pi_0 \mapsto t, \pi_1\mapsto f(t)] \models \aut$ for all $t \in \traces{\tsys}$.
Such a function \myquote{explains} why $\tsys$ satisfies $\varphi$.

Dually, consider a \hylogicvpa formula~$\varphi' = \exists \pi_0.\ \forall \pi_1.\ \aut'$ and a transition system~$\tsys'$ such that $\tsys' \not\models \varphi'$.
Then, there is a function~$f \colon \traces{\tsys'} \rightarrow \traces{\tsys'}$ such that $[\pi_0 \mapsto t, \pi_1\mapsto f(t)] \not\models \aut'$ for all $t \in \traces{\tsys'}$, i.e., $f$ represents a counterexample yielding for every possible trace~$t$ that $\pi_0$ can be mapped to, a trace~$f(t)$ that witnesses that $[\pi_0 \mapsto t] \not\models \forall \pi_1.\ \aut'$. 
Note that $f$ here is a Skolem function for the negation~$\neg\varphi' = \exists \pi_0.\ \forall \pi_1.\ \overline{\aut'}$, where $\overline{\aut'}$ is a $\omega$-VPA accepting the complement of the language of $\aut'$.

Note that this reasoning generalizes to arbitrary formulas in $\Sigma_2$ and $\Pi_2$, i.e., with more than one quantifier in each block. 
Our goal is to compute finite representations of such explanations and counterexamples using a game-based approach initially introduced for \hylogicvpa~\cite{hyperliveness,BF,prophy}.
As already argued above, it suffices to consider explanations for $\Pi_2$ formulas, as counterexamples for $\Sigma_2$ formulas are dual.
Also, we focus on formulas of the form~$\forall \pi.\ \exists \pi'.\ \aut$ (i.e., with a single variable in each quantifier block) and with variables $\pi$ and $\pi'$ instead of $\pi_0$ and $\pi_1$. Both assumptions simplify our notation in the following proof. The renaming of variables is inconsequential while we comment on how to generalize the proof to general $\forall^*\exists^*$ formulas in \cref{remark_generalcase}.

We use pushdown automata with output to compute, given a trace~$t$ for $\pi$, a trace~$f(t)$ as above (cf.\ the introduction of \cite{prophy} for a general introduction to the implementation of Skolem functions by automata with outputs).
In general, such automata will also need some (restricted form of) lookahead on $t$, which we will motivate next.

Intuitively, we capture the semantics of $\tsys \models \forall \pi\ \exists \pi'.\ \aut$ by a two-player game between Falsifier (constructing a trace~$t$ of $\tsys$ for $\pi$) and Verifier (constructing a trace $t'$ of $\tsys$ for $\pi'$). Verifier wins if $\combine{t,t'} \in L(\aut)$.
To obtain decidability of the game, the players need to pick their traces in alternation. 
But this puts Verifier at a disadvantage as she has only access to a prefix of $t$ when determining $t'$, while $t'$ may need to depend on all letters of $t$.

Assume, e.g., that $\aut$ is equivalent to the \ltl formula~$a_{\pi'} \leftrightarrow \F a_{\pi}$, i.e., Verifier needs to pick an $a$ in the first round if and only if Falsifier plays an $a$ in some round. 
Verifier does not have a winning strategy, even though $\tsys$ may satisfy the formula.
However, a single bit of information about Falsifier's move (\myquote{will $t$ contain an $a$?}) is sufficient for Verifier to win.

We define a game where Falsifier, in every round, has to make binding predictions about the membership of the suffix of $t$ starting in the current round for a precomputed list of languages (that only depends on $\aut$ and $\tsys$). 
Such \myquote{prophecies} have previously been applied to \hyltl model-checking with a single~\cite{BF} and any number~\cite{prophy} of quantifier alternations.
In the context-free setting, the prophecies need to give Verifier also information about the evolution of the stack height (which is fully controlled by Falsifier, i.e., he can indeed make predictions about it).

For example, assume that $\aut$ accepts $\combine{t,t'}$ if and only if either
\begin{itemize}
    \item $t$ has a nonempty prefix that causes the stack height of a run of $\aut$ on $\combine{t,t'}$ to reach stack height zero (which only depends on $t$ due to our assumption on control of $\aut$) and $t'(0) = t(n)$, where $n > 0$ is the minimal position with this property, or
    \item $t$ does not have such a prefix and $t'(0) = \#$ for some special symbol~$\#$.
\end{itemize}
This specification does require Verifier not only to have information about the evolution of the stack height (which is under the control of Falsifier) but also about which letter Falsifier is playing when reaching stack height zero for the first time again.
In general, one can even construct examples where Verifier needs information about moves at the next time the current stack height (which may possibly be nonzero) is reached again for the first time.

In the prophecies we define later, we will actually require Falsifier to provide even more information: he does not only have to provide information about (certain) future letters, but also with which transition they may be processed.

Our main result is that for every $\omega$-VPA~$\aut$ and transition system~$\tsys$, there is a computable list of prophecies so that Verifier wins the game with these prophecies if and only if $\tsys \models \forall \pi.\ \exists \pi'.\ \aut$, and the resulting game can be solved effectively.
Furthermore, (finite representations of) winning strategies can be computed effectively, which yield the desired explanations for  $\tsys \models \forall \pi.\ \exists \pi'.\ \aut$.

This section is structured as follows: In \cref{subsec_gsgames}, we introduced standard two-player games and in \cref{subsec_gameswithprophecies} explain how to use them to describe games with prophecies. Then, in \cref{subsec_prophecies}, we introduce the prophecies we need in our construction and state the correctness of our construction.
Then, in \cref{subsec_soundness} and \cref{subsec_completeness}, we prove correctness.

\subsection{Gale-Stewart Games.}
\label{subsec_gsgames}

A Gale-Stewart game~$\gsgame(L)$ is given by an $\omega$-language~$L \subseteq (\SigmaI \times \SigmaO)^\omega$, its winning condition.
It is played between Falsifier and Verifier in rounds~$i = 0,1,2, \ldots$:
In each round, first Falsifier picks a letter~$\alpha(i) \in \SigmaI$, then Verifier picks a letter~$\beta(i) \in \SigmaO$.
After $\omega$ rounds, the players have constructed an outcome~$\combine{\alpha(0)\alpha(1)\alpha(2)\cdots, \beta(0)\beta(1)\beta(2)\cdots}$
which is winning for Verifier if it is in $ L$.
A strategy for Verifier in $\gsgame(L)$ is a mapping~$\sigma \colon \SigmaI^* \rightarrow \SigmaO$. 
An outcome as above is consistent with $\sigma$, if $\beta(i) = \sigma(\alpha(0) \cdots \alpha(i))$ for all $i$. 
A strategy~$\sigma$ for Verifier is winning if every outcome that is consistent with $\sigma$ is in $L$. 
Verifier wins $\gsgame(L)$ if she has a winning strategy for $\gsgame(L)$.

\begin{proposition}[\cite{lms}]
\label{prop_games}
    The following problem is \twoexp-complete: Given an $\omega$-VPA~$\aut$ over a product alphabet~$\SigmaI \times \SigmaO$, does Verifier win $\gsgame(L(\aut))$?
    Furthermore, if Verifier wins $\gsgame(L(\aut))$, winning strategies implemented by pushdown automata with output can be computed in doubly-exponential time.
\end{proposition}

Löding et al.\ formally proved that the winner (and winning strategies implemented by pushdown automata with output) of games played on configuration graphs of visibly pushdown systems (visibly pushdown automata without an acceptance condition) and with winning conditions given by $\omega$-VPA can be determined in doubly-exponential time. Gale-Stewart games with winning conditions given by $\omega$-VPA can be reduced to the games considered by Löding et al.\ with a linear blowup. 

\subsection{Games with Prophecies.}
\label{subsec_gameswithprophecies}

For the remainder of the section, fix a transition system~$\tsys$ with set~$V$ of vertices and a formula~$\phi = \forall \pi.\ \exists \pi'.\ \aut$, such that $\aut$ is controlled by the first trace. 
Furthermore, let $\Sigma$ be the alphabet used to label vertices in $\tsys$, i.e., the alphabet of $\aut$ is $\Sigma \times \Sigma$. 
As $\aut$ is an $\omega$-VPA, $\Sigma \times \Sigma$ is partitioned into $(\Sigmacall', \Sigmareturn', \Sigmaskip')$. Finally, as $\aut$ is controlled by the first component, there is a partition of $\Sigma$ into $(\Sigmacall, \Sigmareturn, \Sigmaskip)$ such that $\Sigmacall' = \Sigmacall \times \Sigma$, $\Sigmareturn' = \Sigmareturn \times \Sigma$, and $\Sigmaskip' = \Sigmaskip \times \Sigma$.
We call $(\Sigmacall, \Sigmareturn, \Sigmaskip)$ the projected partition of $\aut$.

To make the predictions of Falsifier indeed binding, we use so-called prophecy variables that are used by Falsifier to make his predictions and then use the winning condition of the game to ensure he loses when he violates a prediction. 
Formally, let $\mathcal{P} $ be a finite set of prophecies (to be defined later), i.e., languages over $\Sigma$, and let $\pvar{P}$ be the prophecy variable associated to $P \in \cal{P}$.

Now, let $\SigmaI = V \times 2^{\set{\pvar{P} \mid P \in \mathcal{P}}}$, let $\SigmaO = V$, let $\projone{\cdot}$ and $\projtwo{\cdot}$ denote the projections from $\SigmaI$ to $V$ and from $\SigmaI$ to $2^{\set{\pvar{P} \mid P \in \mathcal{P}}}$, and let $L(\aut,\tsys,\mathcal{P}) $ be the language
\begin{align*}
&\set{
\combine{\alpha, \beta} \mid  \alpha\in \SigmaI^\omega, \beta\in\SigmaO^\omega, \text{ if $\projone{\alpha}$ is a path of $\tsys$, then $\beta$ is a path of $\tsys$ and},\\
&\qquad\left[ 
\forall i \in \nats.\ \forall P \in \mathcal{P}.\ \pvar{P} \in \projtwo{\alpha(i)} \leftrightarrow \lambda(\projone{\alpha(i)\alpha(i+1)\alpha(i+2)\cdots}) \in P 
\right] \rightarrow\\
&\qquad\qquad\combine{\lambda(\projone{\alpha}),\lambda(\beta)} \in L(\aut)
}    
\end{align*}
expressing that whenever Falsifier's predictions are correct, then the specification expressed by $\aut$ must be satisfied.
Further, we require Falsifier to actually pick a path in $\tsys$. If he does so, then Verifier also needs to pick a path in $\tsys$, otherwise she trivially wins.

\subsection{Prophecies for HyperVPA Model-Checking.}
\label{subsec_prophecies}
In the following, for $\tsys$ and $\aut$ as fixed above, we present a finite set~$\mathcal{P}$ of prophecies such that $\tsys \models \forall \pi.\ \exists \pi'.\ \aut$ if and only if Verifier wins $\gsgame(L(\aut,\tsys,\mathcal{P}))$.
Further, we show that $L(\aut,\tsys,\mathcal{P})$ is recognized by an $\omega$-VPA, i.e., the winner of the game can be effectively determined.

We continue by recalling some elementary properties about the evolution of the stack during a run of a pushdown automaton.

\begin{figure}[tb]
    \centering
    \usetikzlibrary{decorations.pathreplacing}
\usetikzlibrary{decorations.pathmorphing}
\usetikzlibrary{shapes,shapes.symbols,automata,arrows,shadows}
\usetikzlibrary{calc}
\usetikzlibrary{math} 
 
\tikzset{p0/.style = {shape = circle,    draw, thick, minimum size = 0.4cm}}
\tikzset{p1/.style = {shape = rectangle, draw, thick, minimum size = 0.4cm}}
\tikzset{>=stealth, shorten >=1pt}
\tikzset{every edge/.style = {thick, ->, draw}}
\tikzset{every loop/.style = {thick, ->, draw}}
 
\hspace*{-0.275cm}
\begin{tikzpicture}[scale=.92]

\tikzmath{\i = .18; \o =.24;} 

\draw[fill=black] (0,0)    circle (\i cm);      \draw[thick, draw=black, fill=none] (0,0) circle (\o cm);
\draw[fill=black] (1,0.5)  circle (\i cm);      \draw[thick, draw=black, fill=none] (1,0.5) circle (\o cm);
\draw[fill=black] (2,1)    circle (\i cm);     \draw[thick, draw=black, fill=none] (2,1) circle (\o cm);
\draw[fill=black] (3,1.5)  circle (\i cm);
\draw[fill=black] (4,2)    circle (\i cm);
\draw[fill=black] (5,1.5)  circle (\i cm);
\draw[fill=black] (6,2)    circle (\i cm);       
\draw[fill=black] (7,1.5)  circle (\i cm);     
\draw[fill=black] (8,1)    circle (\i cm);     \draw[thick, draw=black, fill=none] (8,1) circle (\o cm);
\draw[fill=black] (9,1.5)  circle (\i cm);     \draw[thick, draw=black, fill=none] (9,1.5) circle (\o cm);
\draw[fill=black] (10,1.5) circle (\i cm);     \draw[thick, draw=black, fill=none] (10,1.5) circle (\o cm);
\draw[fill=black] (11,2)   circle (\i cm);      
\draw[fill=black] (12,2.5) circle (\i cm);
\draw[fill=black] (13,2)   circle (\i cm);
\draw[fill=black] (14,1.5) circle (\i cm);     \draw[thick, draw=black, fill=none] (14,1.5) circle (\o cm);
 
\path[draw, thick] (0,0) -- (1,0.5) -- (2,1) -- (3,1.5) -- (4,2) -- (5,1.5) -- (6,2) -- (7,1.5) -- (8,1) -- (9,1.5) -- (10,1.5) -- (11,2) -- (12,2.5) -- (13,2) -- (14,1.5);
 
\foreach \x in {0, 0.5,...,2.5}
\path[draw, dashed] (-.6, \x) -- (14,\x);
 
 
\path[draw, -stealth, thick] (-.6,0) -- (-.6,3.0);
 
 
 
\draw[draw=none] (-0.8,0) -- node[sloped]{stack height} (-0.8,3);

\node[draw=none] at (0,0) {\footnotesize\color{white}0};
\node[draw=none] at (1,0.5) {\footnotesize\color{white}1};
\node[draw=none] at (2,1) {\footnotesize\color{white}2};
\node[draw=none] at (3,1.5) {\footnotesize\color{white}3};
\node[draw=none] at (4,2) {\footnotesize\color{white}4};
\node[draw=none] at (5,1.5) {\footnotesize\color{white}5};
\node[draw=none] at (6,2) {\footnotesize\color{white}6};
\node[draw=none] at (7,1.5) {\footnotesize\color{white}7};
\node[draw=none] at (8,1) {\footnotesize\color{white}8};
\node[draw=none] at (9,1.5) {\footnotesize\color{white}9};
\node[draw=none] at (10,1.5) {\footnotesize\color{white}10};
\node[draw=none] at (11,2) {\footnotesize\color{white}11};
\node[draw=none] at (12,2.5) {\footnotesize\color{white}12};
\node[draw=none] at (13,2) {\footnotesize\color{white}13};
\node[draw=none] at (14,1.5) {\footnotesize\color{white}14};

\end{tikzpicture}
    \caption{Development of the stack height during a run.
        Positions that are steps (that is, onward, the stack height never goes below the current stack height) are circled, all other positions lie inside humps.
        The pair~$(2,7)$ of positions  is a call matched by a return. Other such pairs are $(3,4)$, $(5,6)$, $(10,13)$, and $(11,12)$.
        Positions~$0$, $1$, and $8$ are calls which remain unmatched.
        The result of an unmatched call is that onward, the current stack height is never reached again. Hence, these positions are proper steps.
    }
    \label{fig_pushdownstack}
\end{figure}

\begin{remark}[Step, hump, matching]
\label{remark_steps}
A \emph{step} in a run is a position $n$ such that from there onward, the stack height will never be below the stack height at position $n$. Every run has infinitely many steps.
A step is proper, if its stack height is never reached again. 
Whenever, a position of a run is not a step, we say the run is in a \emph{hump}. Between two steps of the same stack height, the stack will be build up and down, giving the appearance of a \myquote{hump} over time.
A call is \emph{matched} if the stack height at the call is reached again. It is then matched with the return that reaches that stack height again for the first time. Otherwise, the call is \emph{unmatched}.
\cref{fig_pushdownstack} shows an illustration. 
\end{remark}

A central role of our prophecies is to gather information about the development of the stack height during a run.
Let us stress again that the stack height during a run is controlled solely by Falsifier, as this enables the approach of model-checking via games with prophecies in the setting of visibly pushdown specifications: The intention of prophecies is that Falsifier has to give Verifier additional information that help her overcome the disadvantage of having to pick~$t'$ without knowing $t$ completely.
This requires that Falsifier only makes truthful predictions about the future.
If Falsifier is not truthful, Verifier wins automatically.
Assume that Falsifier makes predictions about the development of the stack height. 
If Verifier's moves can also influence the stack height, she can actively render Falsifier's predictions false, making her win the game unjustly. 

Before introducing the prophecies concerning the development of the stack height, we need some auxiliary notions.
Given a finite word~$w$ over $\Sigma = \Sigmacall \cup \Sigmareturn \cup \Sigmaskip$ (recall that $(\Sigmacall,\Sigmareturn,\Sigmaskip)$ is the projected partition), we define its effect~$\effect(w) \in \ints$ as $|w|_{\Sigmacall} - |w|_{\Sigmareturn}$.
Note that this does not necessarily match the stack height of a run processing~$w$, as a return on an empty stack height is like a skip. However, if the effect never gets negative, then it matches.

In the following, we define our prophecies and give some intuition, see also \cref{fig_pushdownstack}.
We begin with two auxiliary prophecies.
\begin{itemize}
\itemsep0.5em
\item 
$\step = \set{ t \in \Sigma^\omega \mid \effect(t(0) \cdots t(n)) \ge 0 \text{ for all } n \ge 0} $  
\begin{itemize}
    \item 
With this prophecy, we require Falsifier to predict whether the current position is a step (i.e., the current stack top symbol (and everything below) will never be removed). 
\end{itemize}

\item
$\properstep = \set{ t \in \Sigma^\omega \mid \effect(t(0) \cdots t(n)) \ge 1 \text{ for all } n \ge 0} $
\begin{itemize}
    \item 
Here, we require Falsifier to say whether the current position is a \emph{proper step}, i.e., whether it is a step and the current stack height is never reached again. Note that this implies that $t(0)$  can only be processed by a call-transition. 
\end{itemize}
\end{itemize}

Before we can introduce our main prophecies (which are to be used in conjunction with $\step$ and $\properstep$), we need two more definitions.

\begin{definition}[Fresh run]
Let $r = c_0\tau_0c_1\tau_1c_2\tau_2\cdots$ be a run.
We say that $r$ is \emph{fresh} if
\begin{itemize}
    \item $c_0 = (q, \bot A)$ for some $A \neq \bot$ and $\tau_0 = (q,A,a,q',\gamma)$,  or 
    \item $c_0 = (q, \bot)$ and $\tau_0 = (q,\bot,a,q',\gamma)$,
\end{itemize}
i.e., the stack content of $c_0$ is the smallest one that enables $\tau_0$.  
\end{definition}

\begin{remark}
Fix some $t,t' \in \Sigma^\omega$ and let $\rho$ be a run of $\aut$ starting in some (not necessarily initial) configuration~$(q,\gamma A)$ processing $\combine{t,t'}$.
If the effect of $t(0) \cdots t(n)$ is nonnegative for all $n$, then every configuration reached during $\rho$ has a stack content of the form~$\gamma A \gamma'$, i.e., $\gamma A$ is never removed from the stack and the only symbol that is ever \myquote{accessed}  from the pre-filled part of the stack is the initial top stack symbol~$A$.
Thus, one can analyze runs processing $\combine{t,t'}$ by considering fresh runs only.
This observation is crucial for the construction of our prophecies.

On the other hand, if the effect of some $t(0) \cdots t(n)$ is negative, then we will only reason about prefixes with a nonnegative effect.
\end{remark}

Also, we need to express that we can build an accepting run of the $\omega$-VPA $\aut$.
Since $\aut$ uses Büchi acceptance, an accepting run is a run that visits the set of accepting states infinitely often.
However, it is not enough to specify that it is possible in the future to visit infinitely many times an accepting state, a run must actually make progress towards visiting an accepting state in order to be accepting in the limit.
Hence, we need to compare runs in terms of visiting an accepting state as soon as possible.

\begin{definition}[$\firstF{\cdot}$]
Given a (possibly finite) run $r = c_0\tau_0c_1\tau_1c_2\tau_2\cdots$ of $\aut$, we define $\firstF{r} = n$ where $n$ is minimal with $c_n$ being a configuration whose state is accepting. If no such configuration exists, we define $\firstF{r} = \infty$.
\end{definition}

Now, we introduce the main prophecies.
The first ones are relevant whenever Falsifier indicates that the current position (in a run) is a step (as indicated by the prophecy \step{}).
We dub them \myquote{\texttt{Step}}-prophecies.
Assume Falsifier has already made his move in some round, say to vertex~$u$, i.e., $\lambda(u)$ determines what type (that is, call, skip or return) the transition processing~$\lambda(u)$ and Verifier's move in $\aut$ has.
If the current position is a step, four possibilities can occur:
If $\lambda(u)$ induces a call, it is either eventually matched by a return (handled via \myquote{\texttt{MatchedCallStep}}) or never matched by a return (handled via \myquote{\texttt{UnmatchedCallStep}}).
Alternatively, $\lambda(u)$ can induce a skip (handled via \myquote{\texttt{SkipStep}}).
Finally, it is also possible that $\lambda(u)$ induces a return.
However, then it must be a return on the empty stack (which induces the same stack behavior as a skip), otherwise the position is clearly not a step.
This is handled via \myquote{\texttt{UnmatchedReturnStep}}.

Each prophecy includes a detailed description of its intended use.
To denote parameters (for any prophecy) we use $u,v,v_\tau,\ldots$ for vertices of $\tsys$, and $\tau,\eta,\ldots$ for transitions of $\aut$.

\begin{itemize}
        \item 
$\prophyMCallAcc{u,v,v_\tau,v_\eta}{\tau,\eta}$ $=$ $\{ t \in \traces{\tsys_u} \mid \text{there exists an infinite path } v_0v_1\cdots$ in $\tsys$ and a fresh infinite run $r = c_0\tau_0c_1\tau_1\cdots$ of $\aut$ on $\combine{t,\lambda({v_0v_1\cdots})}$ such that:
{\setlength{\abovedisplayskip}{0pt}%
\setlength{\abovedisplayshortskip}{0pt}%
    \begin{align*}
        & v_0 = v_\tau \in \Succ{v}, \text{ and } \tau_0 = \tau, \text{ and $\tau$ is a call-transition}, \\
        & n\in \nats \text{ is minimal with } \sh(c_0) = \sh(c_{n+1})\text{ (which we require to exist)},\\
        & v_{n} = v_\eta, \text{ and } \tau_n = \eta \text{ (note that $\eta$ must be a return-transition}),\\
        & r \text{ is accepting},\text{ and } \\
        & \text{for all paths } \rho' \text{ in } \tsys_{v'} \text{ where $v' \in \Succ{v}$ and for all accepting runs } r' \text{ of } \aut \\
        & \text{\quad on } \combine{t,\lambda(\rho')} \text{ of the form } c_0\tau'_0c'_1\tau'_1\cdots \text{ we have } \firstF{r} \leq \firstF{r'}\}.
    \end{align*}}
    
    \begin{itemize}
        \item Falsifier continues his path from $u$, and Verifier continues her path from $v_\tau \in \Succ{v}$, the traces induced by these continuations are called $t$ and $t'$, respectively.
        \item The above starting point reflects that in the game, first Falsifier makes his move (he has moved to $u$), then it is Verifier's turn, she is in $v$, and needs to move to some vertex in $\Succ{v}$.
        This prophecy reflects her possibilities if she chooses to move to $v_\tau \in \Succ{v}$.
        \item Say, so far $\aut$ has reached $q$ and $A \in \Gammabot$ is the topmost stack symbol (meaning the exact configuration is of the form $(q,\gamma'A)$ for some $\gamma' \in \Gammabot^*$). Note that $q$ and $A$ are not parameters of the prophecy, we just use them for illustrative purposes. 
        \item We require that Verifier has the possibility to continue her trace such that the run on $\combine{t,t'}$ is accepting using $\tau$ as its next transition.
        The current stack top symbol and everything below will never be removed, so having an accepting run continuing from $(q,\gamma'A)$ is equivalent to having an accepting run from $(q,A)$, i.e., to have a fresh run.
        \item As this is a matched call, it is required that its matching return is processed by $\eta$ (processing the letter $(t(n),\lambda(v_\eta)) \in \Sigma^2$) that happens $n$ transitions later.
        \item As explained before the definition of $\firstF{\cdot}$, for acceptance it is important to make actual progress towards visiting an accepting state.
        Here, it is required that $r$ is a run that visits an accepting state as soon as possible among all accepting runs.
        \item In \cref{fig_pushdownstack}, this type of prophecy is used at positions~$2$ and $10$.
    \end{itemize}
\end{itemize}
The above prophecy is central, in a global view, to ensuring Verifier can win the game if the formula is satisfied: 
It is used when the current position is a step and a call occurs that will be matched by a return eventually, i.e., the position after the matching return is a step again.
The prophecy ensures that after the hump between the call and its matching return, Verifier can continue to play in a way that $\aut$ can still accept.

The upcoming three prophecies cover the possibilities to reach a step directly after a step without a hump in between, which simplifies their definitions slightly. 

\begin{itemize}
\itemsep0.5em 
\item
    $\prophyUCallAcc{u,v,v_\tau}{\tau}$ $=$ $\{ t \in \traces{\tsys_u} \mid \text{there exists an infinite path } v_0v_1\cdots$  in $\tsys$  and a fresh infinite run $r = c_0\tau_0c_1\tau_1\cdots$ of $\aut$ on $\combine{t,\lambda({v_0v_1\cdots})}$ such that:
    {\setlength{\abovedisplayskip}{0pt}%
\setlength{\abovedisplayshortskip}{0pt}%
    \begin{align*}
        & v_0 = v_\tau \in \Succ{v}, \text{ and }\tau_0 = \tau, \text{ and $\tau$ is a call-transition}, \\
        & \sh(c_0) < \sh(c_{i+1}) \text{ for all } i\ge 0,\\
        & r \text{ is accepting}, \text{ and } \\
       & \text{for all paths } \rho' \text{ in } \tsys_{v'} \text{ where $v' \in \Succ{v}$ and all accepting runs } r' \text{ of } \aut 
        \\
        & \text{\quad on } \combine{t,\lambda(\rho')} \text{ of the form } c_0\tau'_0c'_1\tau'_1\cdots \text{ we have } \firstF{r} \leq \firstF{r'}\}. 
    \end{align*}}
    \begin{itemize}
    \item This is very similar to the above case, except that the call is unmatched.
    \item In \cref{fig_pushdownstack}, this type of prophecy is used in positions~$0$, $1$, and $8$.
    \end{itemize}

\item
$\prophySkipAcc{u,v,v_\tau}{\tau}$ $=$ $\{ t \in \traces{\tsys_u} \mid \text{there exists an infinite path } v_0v_1\cdots$  in $\tsys$  and a fresh infinite run $r = c_0\tau_0c_1\tau_1\cdots$ of $\aut$ on $\combine{t,\lambda({v_0v_1\cdots})}$ such that:
{\setlength{\abovedisplayskip}{0pt}%
\setlength{\abovedisplayshortskip}{0pt}%
    \begin{align*}
        & v_0 = v_\tau, \text{ and } \tau_0 = \tau, \text{ and $\tau$ is a skip-transition}, \\
        & r \text{ is accepting},\text{ and } \\
        & \text{for all paths } \rho' \text{ in } \tsys_{v'} \text{ where $v' \in \Succ{v}$ and all accepting runs } r' \text{ of } \aut \\
        & \text{\quad on } \combine{t,\lambda(\rho')} \text{ of the form } c_0\tau'_0c'_1\tau'_1\cdots \text{ we have } \firstF{r} \leq \firstF{r'}\}. 
    \end{align*}}
\begin{itemize}
    \item Exactly like $\prophyUCallAcc{u,v,v_\tau}{\tau}$, but the first letter is a skip, not a call.
    \item In \cref{fig_pushdownstack}, this type of prophecy is used in position~$9$.
\end{itemize}

\item
$\prophyReturnAcc{u,v,v_\tau}{\tau}$ $=$ $\{ t \in \traces{\tsys_v} \mid \text{there exists an infinite path }  v_0v_1\cdots$  in $\tsys$  and a fresh infinite run  $r = c_0\tau_0c_1\tau_1\cdots$ of $\aut$ on $\combine{t,\lambda({v_0v_1\cdots})}$ such that:
{\setlength{\abovedisplayskip}{0pt}%
\setlength{\abovedisplayshortskip}{0pt}%
    \begin{align*}
        & v_0 = v_\tau, \tau_0 = \tau, \text{ and $\tau$ is a return-transition of the form } (q_1,\bot,(\lambda(u),\lambda(v_0)),q_2,\bot), \\
        & r \text{ is accepting},\text{ and } \\
        & \text{for all paths } \rho' \text{ in } \tsys_{v'} \text{ where $v' \in \Succ{v}$ and all accepting runs } r' \text{ of } \aut \\
        & \text{\quad on } \combine{t,\lambda(\rho')} \text{ of the form } c_0\tau'_0c'_1\tau'_1\cdots \text{ we have } \firstF{r} \leq \firstF{r'}\}. 
    \end{align*}}
\begin{itemize}
    \item Like \myquote{\texttt{SkipStep}}, but for the special case of returns on the empty stack.
\end{itemize}
\end{itemize}

The next three prophecies are relevant whenever the current position is not a step, i.e., the position is inside a hump.
We dub them \myquote{\texttt{Hump}}-prophecies.
These prophecies give only information about the future until the stack height goes below the current stack height (which eventually happens, as the position is in a hump, see \cref{remark_steps}).
In a hump, it is possible to encounter calls, which must be matched as unmatched calls always induce a step. This is handled by \myquote{\texttt{CallHump}}.
It is furthermore possible to encounter skips, which is handled by \myquote{\texttt{CallSkip}}, or to encounter returns, which is handled by \myquote{\texttt{ReturnHump}}.

The \myquote{\texttt{Step}}-prophecies take a global view;
their purpose is to ensure that in every position that is a step, Verifier can still win taking the whole future into account.
In contrast, the \myquote{\texttt{Hump}}-prophecies take on a more local view;
their purpose is to ensure that Verifier is able to handle the current hump. 

\begin{itemize}
\itemsep0.5em
    \item 
$\prophySkip{u,v,v_\tau,v_\eta}{\tau,\eta}$ $=$ $\{ t \in \traces{\tsys_u} \mid$ there exists a finite path $\rho = v_0\cdots v_{n}$  in $\tsys$ and a fresh finite run $r = c_0\tau_0c_1\tau_1\cdots\tau_nc_{n+1}$ of $\aut$ on $\combine{t[0,n],\lambda(\rho)}$ such that:
{\setlength{\abovedisplayskip}{0pt}%
\setlength{\abovedisplayshortskip}{0pt}%
    \begin{align*}
        & v_0 = v_\tau, \text{ and } \tau_0 = \tau, \text{ and $\tau$ is a skip-transition}, \\
        & n \in\nats \text{ is minimal with } \sh(c_0) - 1 = \sh(c_{n+1}) \text{ (which we require to exist)}\},\\
        & v_n = v_\eta, \text{ and } \tau_n = \eta \text{ (note that $\eta$ must be a return-transition}),\text{ and } \\
        & \text{for all paths } \rho' = v'_0\cdots v'_n \text{ in $\tsys$ where } v'_0 \in \Succ{v} \text { and } v'_n = v_n \text{ and }\\
        &\text{\quad all finite runs } r' \text{ of } \aut \text{ on } \combine{t[0,n],\lambda(\rho')} \text{ of the form } c_0\tau'_0c'_1\tau'_1\cdots c'_{n}\tau_nc_{n+1}\\
        &\text{\quad we have } \firstF{r} \leq \firstF{r'}\}. 
    \end{align*}}
    \begin{itemize}
        \item The current stack top symbol (say, for illustration purposes only, it is $A$ and the exact configuration reached is of the form $(q,\gamma A)$ for some $\gamma \in \Gammabot^*$) will be removed eventually (note that this implies $A \neq \bot$).
        \item We are now only interested in what happens until the current stack top symbol is removed (what happens after is handled by some \myquote{\texttt{MatchedCallStep}}-prophecy).
        \item Concretely, we require that Verifier has the possibility that the run on $\combine{t,t'}$ continues from $(q_1,\gamma A)$ using $\tau$ (processing $(t(0),\lambda(v_\tau)) \in \Sigma^2$) as its next transition and the current stack top symbol $A$ is removed via $\eta$ (processing $(t(n),\lambda(v_\eta)) \in \Sigma^2$) after $n$ transitions.
        \item The concrete stack content $\gamma$ below $A$ is irrelevant for the calculation of the run until the current top symbol is removed, thus it suffices to specify that there is such a run starting from $(q_1,A)$ instead, that is, it suffices to consider such a fresh run.
        \item As for the other prophecies, since we use Büchi acceptance, we need to make progress towards visiting an accepting state.
        Hence, the finite run $r$ must be a run that visits an accepting state as soon as possible (if possible at all)  compared to all other runs that respect the same conditions on how the next return must happen.
    \end{itemize}
\item
 $\prophyCall{u,v,v_\eta,v_\tau,v_\vartheta}{\tau,\eta,\vartheta}$ $=$ $\{ t \in \traces{\tsys_u} \mid$ there exists a finite path $\rho = v_0\cdots v_{n}$  in $\tsys$ and a fresh finite run $r = c_0\tau_0c_1\tau_1\cdots\tau_nc_{n+1}$ of $\aut$ on $\combine{t[0,n],\lambda(\rho)}$ such that:
 {\setlength{\abovedisplayskip}{0pt}%
\setlength{\abovedisplayshortskip}{0pt}%
    \begin{align*}
        & v_0 = v_\tau, \text{ and } \tau_0 = \tau, \text{ and $\tau$ is a call-transition}, \\
        & m\in\nats \text{ is minimal with } \sh(c_0) = \sh(c_{m+1}) \text{ (which we require to exist)},\\
        & v_m = v_\eta, \text{ and } \tau_m = \eta \text{ (note that $\eta$ must be a return-transition}),\\
        & n \in\nats \text{ is minimal with } \sh(c_0) - 1 = \sh(c_{n+1}) \text{ (which we require to exist)},\\
        & v_n = v_\vartheta, \text{ and } \tau_n = \vartheta \text{ (note that $\vartheta$ must be a return-transition}),\text{ and } \\       
        & \text{for all paths } \rho' = v'_0\cdots v'_n \text{ in $\tsys$ where } v'_0 \in \Succ{v} \text { and } v'_n = v_n\\ 
        &\text{\quad and all finite runs } r' \text{ of } \aut
        \text{ on } \combine{t[0,n],\lambda(\rho')} \text{ of the form } c_0\tau'_0c'_1\tau'_1\cdots c'_{n}\tau_nc_{n+1}\\ &\text{\quad we have } \firstF{r} \leq \firstF{r'}\}. 
    \end{align*}}
    \begin{itemize}
        \item The current letter is a call (obviously matched, as all calls in a hump are). 
        We require that the call is processed by $\tau$ (processing $(t(0),\lambda(v_\tau)) \in \Sigma^2$)) and the matching return, $m$ transitions later, processed by $\eta$ (processing $(t(m),\lambda(v_\eta)) \in \Sigma^2$).
        \item Since we are already in hump, the stack symbol that was on top at the beginning will be removed too, after $n-m$ additional transitions, via $\vartheta$ (processing $(t(n),\lambda(v_\vartheta)) \in \Sigma^2$).
        \item Again, we need to make progress towards visiting an accepting state.
        As for the previous prophecy, how the topmost stack symbol must be removed has been fixed before, when its matching call has been made. 
        Thus, we are looking for an optimal (in terms of visiting an accepting state as soon as possible) path/run among those that, after $n$ transitions, reach $v_\vartheta$ and use $\vartheta$.
        \item In \cref{fig_pushdownstack}, this type of prophecy is used in positions~$3$, $5$, and $11$.
    \end{itemize}
\item
 $\prophyReturn{u,v_\eta}{\eta}$ $=$ $\{ t \in \traces{\tsys_u} \mid \eta $ has the form $(q,A,(\lambda(u),\lambda(v_\eta)),q',\varepsilon)$
  for some $q,q' \in Q,$ and $A \in \Gamma$ $\}$. 
    \begin{itemize}

        \item The topmost stack symbol is removed and it is specified how Verifier is able to do so.
        \item This prophecy does not include the condition regarding visiting an accepting state as soon as possible.
        This is because all returns (that are not returns on the empty stack) have been fixed beforehand when their matching call was made.
        At that time, progress towards visiting an accepting state has been ensured. 
        This also explains why this prophecy does not have a vertex $v$ (symbolizing the vertex Verifier is in) as parameter, since we do not care that Verifier could also move to some $v' \in \Succ{v}$ with $v' \neq v_\eta\in \Succ{v}$.
        \item This type of prophecy is used in all positions where a return happens on a nonempty stack.
        In \cref{fig_pushdownstack}, that is, positions~$4$, $6$, $7$, $12$, and $13$. 
        \end{itemize}
\end{itemize}

For each of these prophecies, one can construct an $\omega$-VPA accepting it. 
The only nontrivial aspect here is to ensure that the run~$r$ referred to in the definitions of the prophecies satisfies $\firstF{r} \leq \firstF{r'}$ for all other runs~$r'$, which can be taken care of using projection and complementation. Here, we again rely on the stack height being controlled by the universally quantified variable only.

Recall that a fresh run is a run~$r = c_0\tau_0c_1\tau_1c_2\tau_2\cdots$ such that 
\begin{itemize}
    \item $c_0 = (q, \bot A)$ for some $A \neq \bot$ and $\tau_0 = (q,A,a,q',\gamma)$,  or 
    \item $c_0 = (q, \bot)$ and $\tau_0 = (q,\bot,a,q',\gamma)$,
\end{itemize}
Note that $c_0$ is uniquely determined by $\tau_0$. 
Hence, a fresh run is uniquely determined by the sequence~$\tau_0 \tau_1 \tau_2 \cdots$, as each $c_{n+1}$ is uniquely determined by applying $\tau_n$ to $c_n$.
In the following, we will often use this property.

\begin{lemma}
\label{lemma_propheciesarevpls}    
Each prophecy in $\cal{P}$ is recognized by an $\omega$-VPA of at most exponential (in $\size{\aut}$ and $\size{\tsys}$) size. Furthermore, all these $\omega$-VPA use the projected partition of $\aut$.
\end{lemma}

\begin{proof}
We begin by present $\omega$-VPA for the auxiliary prophecies~$\step$ and $\properstep$ in \cref{fig_proph_aux}.

The automaton for $\step$ keeps track of the effect of the prefix processed thus far using the stack height, as long as the effect is nonnegative: All inputs that would yield a prefix with negative effect cannot be processed (i.e., a letter from $\Sigmareturn$ when the stack is empty).

The automaton for $\properstep$ works similarly, but we first require to process a letter from $\Sigmacall$, which puts a single~$B$ on the stack. From there onward, the automaton also keeps track of the effect of the prefix processed thus far using the stack height, as long as the effect does not reach zero again: All inputs that would yield a prefix with effect zero cannot be processed (i.e., a letter from $\Sigmareturn$ when the topmost stack symbol is $B$).

\begin{figure}
    \centering

    \begin{tikzpicture}[thick]

  \node[state,accepting] (qi) at (0,0) {};

  \path[->, > = stealth]
  (-1,0) edge (qi)
  (qi) edge[loop above] node[above,align=center] {$c, \bot \mid \bot A$ \\ $c, A \mid AA$} ()
  (qi) edge[loop right] node[right] {$r, A \mid \epsilon$} ()
  (qi) edge[loop below] node[below,align=center] {$s, \bot \mid \bot$ \\ $s, A \mid A$} ()
 ;

  \node[state] (qii) at (5,0) {};
   \node[state] (qiii) at (8,0) {};

  \path[->, > = stealth]
  (4,0) edge (qii)
  (qii) edge node[above] {$c, \bot \mid \bot B$} (qiii)
  (qiii) edge[loop above] node[above,align=center] {$c, B \mid B A$ \\ $c, A \mid AA$} ()
  (qiii) edge[loop right] node[right] {$r, A \mid \epsilon$} ()
  (qiii) edge[loop below] node[below,align=center] {$s, B \mid B$ \\ $s, A \mid A$} ()
 ;

    \end{tikzpicture}

    \caption{The $\omega$-VPA recognizing $\step$ and $\properstep$. A transition~$(q, X, a, q', \gamma)$ is depicted by an edge from $q$ to $q'$ labeled by $a, X \mid \gamma$. Here, $c$, $r$, and $s$ stand for arbitrary letters from $\Sigmacall$, $\Sigmareturn$, and $\Sigmaskip$, respectively.}
    \label{fig_proph_aux}
\end{figure}

Next, let us consider~$\prophyMCallAcc{u,v,v_\tau,v_\eta}{\tau,\eta}$. We first construct two auxiliary $\omega$-VPA and then use closure properties to obtain an $\omega$-VPA recognizing the prophecy.
To this end, first consider the language~$L_1$ of words of the form~$\combine{t, \rho, \tau_0 \tau_1 \tau_2 \cdots}$ satisfying the following properties (cp.\ the definition of $\prophyMCallAcc{u,v,v_\tau,v_\eta}{\tau,\eta}$):
\begin{itemize}
    \item $t \in \traces{\tsys_u}$,
    \item $\rho = v_0 v_1 v_2 \cdots$ is a path of $\tsys$,
    \item $\tau_0 \tau_1 \tau_2 \cdots$ is a sequence of transitions of $\aut$ that induces a fresh run~$c_0 \tau_0 c_1 \tau_1\cdots $ of $\aut$ on $\combine{t, \lambda(\rho)}$ (which is uniquely determined by $\tau_0 \tau_1 \tau_2 \cdots$),
    \item $v_0 = v_\tau \in \Succ{v}$, and $\tau_0 = \tau$, and $\tau$ is a call-transition,
    \item $n$ is the smallest number such $\sh(c_0) = \sh(c_{n+1})$ (which we require to be well-defined),
    \item $v_n = v_\eta$ and $\tau_n = \eta$ (note, $\eta$ is a return-transition), and 
    \item $r$ is accepting. 
\end{itemize}
It expresses all but the last requirement in the definition of $\prophyMCallAcc{u,v,v_\tau,v_\eta}{\tau,\eta}$. The last one will be taken care of later.

One can construct an $\omega$-VPA recognizing $L_1$ using the product of the states of $\aut$ (to simulate the run~$r$ induced by $\tau_0 \tau_1 \tau_2 \cdots$) and two copies of the vertices of $\tsys$ (to check that $t$ is in $\traces{\tsys_u}$ and to guess the path~$\rho$), and using the same stack alphabet as $\aut$ (again, to simulate $r$). 
Furthermore, it checks all the initial constraints on $v_0$ and $\tau_0$ using the state space and uses the stack to ensure that the first time the stack height~$\sh(c_0)$ is reached, the vertex just processed in $\rho$ is $v_\eta$ and that the transition used to process it is $\eta$. This requires an additional component of the states to keep track of whether that stack height has been reached already or not. The accepting states are inherited from $\aut$ to ensure that the run~$r$ of $\aut$ induced by $\tau_0 \tau_1 \tau_2 \cdots$ is accepting, and that stack height~$\sh(c_0)$ is reached again. 
Finally, the partition is the one induced by the projected one of $\aut$, i.e., it depends only on the first component.
We leave the slightly tedious, but straightforward details to the reader.

Instead, we focus on the last requirement, i.e., that there is no other accepting run~$r'$ starting in $c_0$ and processing $\combine{t, \lambda(\rho')}$ for some path~$\rho'$ starting in some $v_0' \in \Succ{v}$ such that $\firstF{r'} < \firstF{r}$.
To this end, consider the language~$L_2$ of words of the form~$\combine{t, \rho, \tau_0 \tau_1 \tau_2 \cdots, \rho', \tau_0' \tau_1' \tau_2' \cdots}$ satisfying the following properties:
\begin{itemize}
    \item $\rho = v_0 v_1 v_2 \cdots$ is a path of $\tsys$,
    \item $\rho' = v_0' v_1' v_2' \cdots$ is a path of $\tsys$ starting in some $v_0' \in \Succ{v}$,
    \item $\tau_0\tau_1\tau_2$ is a sequence of transitions of $\aut$  that induces a (not necessarily accepting) fresh run~$c_0 \tau_0 c_1 \tau_1\cdots $ of $\aut$ on $\combine{t, \lambda(\rho)}$ (which is uniquely determined by $\tau_0 \tau_1 \tau_2 \cdots$),
    \item $\tau_0'\tau_1'\tau_2'$ is a sequence of transitions of $\aut$ that induces an accepting fresh run~$c_0 \tau_0' c_1' \tau_1'\cdots $ of $\aut$ on $\combine{t, \lambda(\rho')}$ (which is uniquely determined by $\tau_0' \tau_1' \tau_2'\cdots$), and 
    \item $\firstF{r'} < \firstF{r}$.
\end{itemize}

An $\omega$-VPA recognizing $L_2$ can be constructed using the product of two copies of the state space of $\aut$ (to simulate the runs~$r$ and $r'$) and two copies of the set of vertices of $\tsys$ (to check that $\rho$ and $\rho'$ are indeed paths of $\tsys$), and using the product of two copies of the stack alphabet of $\aut$ (to simulate the runs~$r$ and $r'$). Here, we rely on the fact that $\aut$ is controlled by the letters of $t$ only, i.e., $r$ and $r'$ always have the same stack height.
Again, the automaton checks the initial constraint on $v_0'$ using its state space as well as an additional component of the state space to reject when $\firstF{r'} \ge \firstF{r}$.
The accepting states are inherited from the second copy of the states of $\aut$ to ensure that $r'$ is accepting.
Again, the partition is the one induced by the projected one of $\aut$, i.e., it depends only on the first component.
And again, we leave the details to the reader.

Now, when projecting away the last two components of $L_2$ and then complementing the resulting $\omega$-VPA we have eliminated all $\combine{t, \rho, \tau_0 \tau_1 \tau_2 \cdots}$ where the run~$r$ induced by $\tau_0 \tau_1 \tau_2 \cdots$ is not among the ones reaching $F$ as soon as possible. Call the resulting language~$L_2'$.
Hence, $\prophyMCallAcc{u,v,v_\tau,v_\eta}{\tau,\eta}$ is the language obtained by taking the intersection of $L_1$ and $L_2'$ and then projecting away the last two components (representing $\rho$ and $\tau_0 \tau_1 \tau_2\cdots$). 
Note that both projections we have used here are actually renamings (projections that preserve the partition into calls, returns, and skips).
Hence, as $\omega$-VPA are closed under, complementation, intersection, and renaming~\cite{DBLP:conf/stoc/AlurM04}, we obtain an $\omega$-VPA for $\prophyMCallAcc{u,v,v_\tau,v_\eta}{\tau,\eta}$, which can can be shown to have at most exponential size in $\size{\aut}$ and $\size{\tsys}$, as the only \myquote{expensive} operation is the complementation.

Using a similar approach, one can construct $\omega$-VPA  of exponential size (in $\size{\aut}$ and $\size{\tsys}$) recognizing the prophecies~$\prophyUCallAcc{u,v,v_\tau}{\tau}$, $\prophySkipAcc{u,v,v_\tau}{\tau}$, and $\prophyReturnAcc{u,v,v_\tau}{\tau}$.
Also, the constructions can be adapted for the prophecies~$\prophySkip{u,v,v_\tau,v_\eta}{\tau,\eta}$ and\newline $\prophyCall{u,v,v_\eta,v_\tau,v_\vartheta}{\tau,\eta,\vartheta}$: here, we are only interested in finite runs~$r$ and $r'$ induced by the sequences~$\tau_0 \tau_1 \tau_2 \cdots$ and $\tau_0' \tau_1' \tau_2' \cdots$ until the first position where the stack height is strictly smaller than that of $c_0$. Note that this position only depends on $t$.
Hence, we again obtain $\omega$-VPA of exponential size in $\size{\aut}$ and $\size{\tsys}$.

Finally, the prophecy~$\prophyReturn{u,v_\eta}{\eta}$ only refers to the first letter of the input~$t$ and is is therefore trivially recognized by an $\omega$-VPA with two states.
\end{proof}

Using the $\omega$-VPA for the prophecies and closure properties of these automata, one can show that the winning condition of our game is also accepted by an $\omega$-VPA.

\begin{lemma}
\label{lemma_winningconditionisvpl}    
The winning condition~$L(\aut,\tsys, \cal{P})$ is recognized by an $\omega$-VPA of triply-exponential size (in $\size{\aut}$ and $\size{\tsys}$).
\end{lemma}

\begin{proof}
In the following, all quantities (e.g., automata sizes) are measured in $\size{\aut}$ and $\size{\tsys}$.

There is a polynomial number of prophecies in $\cal{P}$. 
Due to \cref{lemma_propheciesarevpls}, for every $P \in \cal{P}$, there is an exponentially-sized $\omega$-VPA~$\aut_P$ recognizing $P$, which uses the input alphabet~$\Sigma$ and the projected partition of $\aut$.
The alphabet of $\aut_P$ can be extended to the input alphabet~$(V \times 2^{\set{\pvar{P} \mid P \in \mathcal{P}}}) \times V$ so that an input letter~$((v, S), v')$ is projected to $\lambda(v)$. 
The resulting $\omega$-VPA~$\aut_P'$ has the same size as $\aut_P$.
Similarly, $\aut$ can be extended to the input alphabet~$(V \times 2^{\set{\pvar{P} \mid P \in \mathcal{P}}}) \times V$, obtaining the $\omega$-VPA~$\aut'$. All these extended automata use the same partition that only depends on the vertex in the first component (picked by Falsifier).

These $\aut_P'$ can be combined into an alternating $\omega$-VPA~$\aut_{\mathrm{prem}}$ (see~\cite{altVPA} for formal definitions) of exponential size recognizing
\[\forall i \in \nats.\ \forall P \in \mathcal{P}.\ \pvar{P} \in \projtwo{\alpha(i)} \leftrightarrow \lambda(\projone{\alpha(i)\alpha(i+1)\alpha(i+2)\cdots}) \in P. \]
Intuitively, the alternating automaton reads, for each input letter it processes, which prophecies~$P$ are currently predicted to be true (using the component~$2^{\set{\pvar{P} \mid P \in \mathcal{P}}}$ of the input alphabet) and then spawns a fresh copy of each $\aut_P'$ so that $P$ is predicted to be true and spawns a copy of the dual automaton of each $\aut_P'$ so that $P$ is predicted to be false. 
The dual automaton of $\aut_P'$ accepts the complement language of $L(\aut_P')$ w.r.t.\ the alphabet~$(V \times 2^{\set{\pvar{P} \mid P \in \mathcal{P}}}) \times V$, but again ignores all inputs but the first vertex (picked by Falsifier).
Hence, it checks that the prophecy~$P$ is indeed violated.
One needs to take care that each copy spawned at stack height~$h$ treats returns that decrease the stack height below $h$ as returns on the empty stack.
We leave the straightforward details to the reader. 

Now, we complement $\aut_{\mathrm{prem}}$ by dualizing it (without size increase) and take the disjunction with $\aut'$, and a component that checks that if the sequence of vertices picked by Falsifier is a path in $\tsys$, then the sequence of vertices picked by Verifier is a path in $\tsys$ as well.
Altogether, this yields an exponentially sized alternating $\omega$-VPA recognizing $L(\aut, \tsys, \cal{P})$.
This can be turned into a triply-exponential equivalent $\omega$-VPA~\cite{altVPA} recognizing the winning condition.
\end{proof}

Finally, one can show that the game with prophecies in $\cal P$ captures the model-checking problem. 
One direction of the equivalence is rather straightforward:  a winning strategy for Verifier in $\gsgame(L(\aut,\tsys, \cal{P}))$ can be turned into a Skolem function for $\pi'$ witnessing that $\tsys \models \forall \pi.\ \exists \pi'.\ \aut$: given a trace~$t$ of $\tsys$ for $\pi$, consider the play where Falsifier picks the vertices of a path with trace~$t$ and where he always picks the prophecies correctly. Then, the winning strategy yields a $t'$ so that $\combine{t,t'}$ are accepted by $\aut$.

The (much) harder task is to show that the prophecies give Verifier enough information about the trace~$t$ Falsifier is constructing during a play to construct a \myquote{winning} $t'$ without having full access to $t$. 
Intuitively, she always has a move so that from the resulting configuration, she can still win, while also infinitely often ensuring that an accepting state is indeed visited.
Essentially, the parameters (in particular, the $v_\tau$ and the $v_\eta$) of carefully selected prophecies yield such moves.

\begin{lemma}
\label{lemma_prophygameprops}
Let $\cal{P}$ be the set of prophecies introduced above.
Verifier wins $\gsgame(L(\aut,\tsys, \cal{P}))$ if and only if $\tsys \models \forall \pi.\ \exists \pi'.\ \aut$.
\end{lemma}

The proof of the above lemma is presented in the next two subsections:
\begin{itemize}
    \item In Subsection~\ref{subsec_soundness}, we prove the left-to-right direction of the correctness claim (i.e., soundness): If Verifier wins $\gsgame(L(\aut,\tsys, \cal{P}))$ then $\tsys \models \forall \pi.\ \exists \pi'.\ \aut$. 

    \item In Subsection~\ref{subsec_completeness}, we prove the right-to-left direction correctness claim (i.e., completeness): If $\tsys \models \forall \pi.\ \exists \pi'.\ \aut$ then Verifier wins $\gsgame(L(\aut,\tsys, \cal{P}))$. 
\end{itemize}
But, before we present these proofs, let us first state our main result of this section : Combining \cref{lemma_winningconditionisvpl}, \cref{lemma_prophygameprops}, and \cref{prop_games}, we obtain an explanation for $\tsys \models \forall \pi.\ \exists \pi'.\ \aut$ in the form of a Skolem function for $\pi'$, which are implemented by pushdown automata with output. 
Note that these need to have access to the correct values of the prophecy variables at every prefix of the trace assigned to $\pi$.

\begin{theorem}
Explanations for $\Pi_{2,1}$ formulas (in the form of pushdown automata with output and $\omega$-visibly-pushdown lookahead) can be computed in quintuply-exponential time.
\end{theorem}

\begin{remark}
\label{remark_generalcase}
Recall that we restricted ourselves to formulas of the form~$\forall \pi.\ \exists \pi'.\ \aut$, i.e., with a single variable in each quantifier block. 
To generalize the construction to arbitrary $\Pi_{2,1}$ formulas (say with $k$ universal quantifiers and $k'$ existential ones), the prophecies are languages of $k$-tuples of traces and their definition existentially quantifies $k'$ paths. 
So, the parameters~$u$, $v$, $v_\tau$, and $v_\eta$ are replaced by vectors of vertices of length~$k$ or $k'$, respectively.
\end{remark}

Applying \cref{remark_negprops}, we also obtain counterexamples for the dual fragment of $\Pi_{2,1}$.

\begin{corollary}
Counterexamples for $\Sigma_{2,1}$ formulas (in the form of pushdown automata with output and $\omega$-visibly-pushdown lookahead) can be computed in quintuply-exponential time.
\end{corollary}

\subsection{Soundness}
\label{subsec_soundness}

In this subsection, we prove that if Verifier wins $\gsgame(L(\aut,\tsys, \cal{P}))$ then $\tsys \models \forall \pi.\ \exists \pi'.\ \aut$. 
The proof is analogous to the one for game-based model-checking of $\forall^*\exists^*$ \hyltl model-checking~\cite{BF} and independent of the prophecies. 

\begin{proof}
Assume Verifier has a winning strategy for $\gsgame(L(\aut,\tsys,\cal{P}))$. 
To show $\tsys \models \forall \pi.\ \exists \pi'.\ \aut$, we construct for every $t \in \traces{\tsys}$ a $t' \in \traces{\tsys}$ such that $\aut$ accepts $\combine{t,t'}$.

To this end, fix such a $t$ and a path~$\rho$ of $\tsys$ such that $\lambda(\rho) = t$.
Then, consider the outcome~$(\alpha, \beta)$ of $\gsgame(L(\aut,\tsys,\cal{P}))$ where Verifier plays according to her winning strategy and where Falsifier in each round~$i$ picks the vertex~$\rho(i)$ as well as all prophecies in $\cal{P}$ that $t(i)t(i+1)t(i+2)\cdots $ is in. 
This outcome satisfies
\[\forall i \in \nats.\ \forall P \in \mathcal{P}.\ \pvar{P} \in \projtwo{\alpha(i)} \leftrightarrow \lambda(\projone{\alpha(i)\alpha(i+1)\alpha(i+2)\cdots}) \in P\]
by construction.
Hence, as Verifier played according to her winning strategy, which implies that the outcome is winning for her, we can conclude
\[
\combine{\lambda(\projone{\alpha}),\lambda(\beta)} \in L(\aut).
\]
Hence, picking $t' = \lambda(\beta)$ yields the desired result.
\end{proof}

\subsection{Completeness}
\label{subsec_completeness}

In this subsection, we prove that if $\tsys \models \forall \pi.\ \exists \pi'.\ \aut$ then Verifier wins $\gsgame(L(\aut,\tsys, \cal{P}))$.
To this end, we first define a strategy for Verifier (\cref{def_strategy}), and then show that the strategy is winning (\cref{lemma_nonempty,lemma_runAcc}).

During a play of $\gsgame(L(\aut,\tsys, \cal{P}))$, Falsifier picks in each round~$i$ a vertex~$\rho(i)$ of $\tsys$ and a set~$P(i) \subseteq \cal{P}$ of prophecy variables and Verifier then has to pick a vertex~$\rho'(i)$.
We define the strategy for her inductively over the play length~$i$. To so, we inductively compute a subset~$M(i)$ of $P(i)$ and a run prefix~$r = c(0) \tau(0) \cdots \tau(i-1) c(i)$ processing the traces~$t$ and $t'$ of the two sequences~$\rho(0) \cdots \rho(i-1)$ and $\rho'(0) \cdots \rho'(i-1)$ picked by the players thus far (note that $\rho(i)$ is not yet processed, as Verifier has not yet picked $\rho'(i)$). Also note that these sequences are not necessarily path prefixes in $\tsys$, as the rules of the game only require the players to pick vertices.

The $M(i)$ and the run prefix~$r$ are used to determine Verifier's move:
If $M(i)$ is nonempty (we say that Verifier has used the \myquote{\nonempty-case} at position~$i$), then we pick a prophecy~$\ch(i)$ whose prophecy variable is in $M(i)$. This choice depends on~$r$.
The parameters of $\ch(i)$ then determine the move~$\rho'(i)$ the strategy outputs, as well as how to extend~$r$ to process~$(\rho(i), \rho'(i))$.
If $M(i)$ is empty, then we let Verifier pick an arbitrary $\rho'(i)$.
We will later prove that this case never occurs as long as Falsifier makes his predictions correctly. 
If he does not, Verifier wins by definition of $L(\aut, \tsys, \cal{P})$.
So, it then remains to show that Verifier wins all outcomes of plays where she has used the \nonempty-case at all positions. 

For convenience, we say that $\bullet$ is the predecessor of all vertices in $V_\initmark$, i.e., we let $\Succ{\bullet} = V_\initmark$.
This notation is used for round~$i=0$ to instantiate the parameter~$v$ of the prophecies.
Note that this parameter represents the vertex~$\rho'(i-1)$ that Verifier's sequence of moves is currently at, which is undefined for $i=0$. However, the parameter~$v$ in a prophecy is only used to refer to its successors, which for $\bullet$ yields exactly the initial vertices, which are those from which Verifier should select~$\rho'(0)$.

Finally, wether a letter $(a,a') \in \Sigma^*$ is a call, return, or skip in $\aut$ depends only on $a$. Thus, we say the type of $a$ is \myquote{call}, \myquote{return}, or \myquote{skip} and write $\type{a} \in \set{call, return, skip}$.

We are now ready to define the strategy, the $M(i)$, the $\ch(i)$, and the run prefix~$r$.

\begin{definition}[Strategy definition]
\label{def_strategy}
\upshape
Throughout the definition, we assume w.l.o.g.\ that $\aut$ is complete, i.e., in each configuration, every input letter can be processed. This can always be achieved by adding a non-accepting sink state and routing all missing transitions to it. This prevents us having to deal with words that cannot be processed by $\aut$.

We begin our inductive definition in round~$i=0$, i.e., Falsifier has picked a vertex~$\rho(0)$ and a set~$P(0)$ of prophecy variables. Hence, we have $t(0) = \lambda(\rho(0))$.
We define~$c(0) = (q_\initmark,\bot)$, i.e., the run prefix~$r$ consists only of the initial configuration, which does indeed process $\combine{\rho(0)\cdots \rho(i-1), \rho'(0)\cdots \rho'(i-1)}$, which is empty for $i=0$.

Note that the first position of every run is a step, because the stack is empty. We now consider several cases depending on whether Falsifier predicts that it is a proper step and on the type of $t(0)$:
If $\pvar{\properstep} \notin P(0)$ and $\type{t(0)} = call$, then we define
\begin{equation*}
M(0) = \set{\pvar{\prophyMCallAcc{\rho(0),\bullet,v_\tau,v_\eta}{\tau,\eta}} \in P(0) \mid v_\tau,v_\eta \in V,\ \tau,\eta \in \Delta}.
\end{equation*}
If $\pvar{\properstep \in P(0)}$ and $\type{t(0)} = call$, then we define
\begin{equation*}
M(0) = \set{\pvar{\prophyUCallAcc{\rho(0),\bullet,v_\tau}{\tau}} \in P(0) \mid v_\tau \in V,\ \tau \in \Delta}.
\end{equation*}
If $\type{t(0)} = skip$, then we define
\begin{equation*}
M(0) = \set{\pvar{\prophySkipAcc{\rho(0),\bullet,v_\tau}{\tau}} \in P(0) \mid v_\tau \in V,\ \tau \in \Delta}.
\end{equation*}
Lastly, if $\type{t(0)} = return$, then we define
\begin{equation*}
M(0) = \set{\pvar{\prophyReturnAcc{\rho(0),\bullet,v_\tau}{\tau}} \in P(0) \mid v_\tau \in V,\ \tau \in \Delta}.
\end{equation*}
This covers all possible cases.

If $M(0)$ is nonempty, we pick $\ch(0)$ to be an arbitrary prophecy from $M(0)$.
In all four cases above, we define $\rho'(0)$ (i.e., the move selected by the strategy we define) as the parameter~$v_\tau$ of $\ch(0)$ and extend $r$ by the parameter~$\tau$ of $\ch(0)$ and the uniquely determined configuration~$c(1)$ reached by applying $\tau$ in $c(0)$. This is well-defined by the definition of the prophecies.
If $M(0)$ is empty, let $\rho'(0)$ be any $v \in V_\initmark$ and let $\tau(0)$ be any transition of $\aut$ processing $(\lambda(\rho(0)),\lambda(\rho'(0)))$ that is enabled in $c(0)$.

Now, we consider a round~$i > 0$, i.e., Falsifier has picked a sequence $\rho(0) \cdots \rho(i)$ of vertices and a sequence $P(0)\cdots P(i)$ of sets of prophecy variables while Verifier has picked a sequence~$\rho'(0) \cdots \rho'(i-1)$ of vertices and we need to define her move~$\rho'(i)$.
Furthermore, we have defined subsets~$M(j)$ of $P(j)$ for every $j < i$ and a run prefix~$r$ processing the traces~$t$ and $t'$ of $\rho(0) \cdots \rho(i-1)$ and $\rho'(0) \cdots \rho'(i-1)$.

If there is some $j < i$ such that $M(j)$ is empty, then we let $M(i)=\emptyset$, $\rho'(i)$ be any $v \in \Succ{\rho'(i-1)}$, and let $\tau(i)$ be any transition of $\aut$ processing $(\lambda(\rho(i)),\lambda(\rho'(i)))$ enabled in the last configuration of the run prefix~$r$ constructed thus far. 
We then extend $r$ by $\tau(i)$ and the unique configuration reached by applying it in $c(i)$.
Otherwise, we proceed as follows.

First let us consider the case where $\pvar{\step} \in P(i)$ or the last configuration~$c(i)$ of $r$ is has an empty stack.
If additionally $\pvar{\properstep} \notin P(i)$ and $\type{t(i)} = call$, then we define
\begin{equation*}
M(i) = \set{\pvar{\prophyMCallAcc{\rho(i),\rho'(i-1),v_\tau,v_\eta}{\tau,\eta}} \in P(i) \mid v_\tau,v_\eta \in V,\ \tau,\eta \in \Delta}.
\end{equation*}
If additionally  $\pvar{\properstep} \in P(i)$ and $\type{t(i)} = call$, then we define
\begin{equation*}
M(i) = \set{\pvar{\prophyUCallAcc{\rho(i),\rho'(i-1),v_\tau}{\tau}} \in P(0) \mid v_\tau \in V,\ \tau \in \Delta}.
\end{equation*}
If additionally if $\type{t(i)} = skip$, then we define
\begin{equation*}
M(i) = \set{\pvar{\prophySkipAcc{\rho(i),\rho'(i-1),v_\tau}{\tau}} \in P(i) \mid v_\tau \in V,\ \tau \in \Delta}.
\end{equation*}
Lastly, if additionally $\type{t(i)} = return$, then we define
\begin{equation*}
M(i) = \set{\pvar{\prophyReturnAcc{\rho(i),\rho'(i-1),v_\tau}{\tau}} \in P(i) \mid v_\tau \in V,\ \tau \in \Delta}.
\end{equation*}
This again covers all possible subcases.

On the other hand, if $\pvar{\step} \notin P(i)$ and the stack of $c(i)$ is not empty, then we are interested in the maximal $j < i$ where a call-transition is taken in $r$ such that $\sh(c(j)) + 1 = \sh(c(i))$.
For example, in \cref{fig_pushdownstack}, if $i$ is position $3$, then $j$ is position $2$; if $i$ is position $5$, then $j$ is also position~$2$; and if $i$ is position $6$, then $j$ is position $5$.
Then $\ch(j)$ must be of the form 
\begin{equation*}
\prophyCall{\rho(j),\rho'(j-1),\rho'(j),v_\eta,v_\vartheta}{\tau(j),\eta,\vartheta}
\end{equation*}
or
\begin{equation*}
\prophyMCallAcc{\rho(j),\rho'(j-1),\rho'(j),v_\eta}{\tau(j),\eta}
\end{equation*}
for some $v_\eta,v_\vartheta \in V$ and some return-transitions $\eta,\vartheta$.

If the type of $t(i)$ is a call, we are interested in the prophecies that indicate the possibility for Verifier to do the following:
\begin{itemize}
    \item Take some call-transition $\tau$ (processing $(t(i),\lambda(v_\tau)) \in \Sigma^2$) by moving to some vertex $v_\tau \in \Succ{\rho'(i-1)}$.
    \item Say $m$ transitions later, take a return-transition $\eta'$ (processing $(t(i+m),\lambda(v_{\eta'})) \in \Sigma^2$) that will return the call done by $\tau$ by moving to some vertex $v_{\eta'}$.
    \item Say $n-m$ additional transitions later, take a return-transition $\eta$ (processing $(t(i+n),\lambda(v_{\eta})) \in \Sigma^2$) by moving to $v_\eta$ as fixed by $\ch(j)$.
\end{itemize}
Thus, if $\pvar{\step} \notin P(i)$, $\sh(c(i))>0$, and $\type{t(i)} = call$, we define
\begin{align}
\label{eq_callhump}
 M(i) = &\  \{\pvar{\prophyCall{\rho(i),\rho'(i-1),v_\tau,v_{\eta'},v_\eta}{\tau,\eta',\eta}} \in P(i) \mid v_\tau,v_{\eta'}\in V\ \tau,\eta' \in \Delta,\\
 & \quad \text{and $v_\eta$ and $\eta$ as in $\ch(j)$}\}.\nonumber
\end{align}

If $\pvar{\step} \notin P(i)$, $\sh(c(i))>0$, and $\type{t(i)} = skip$, then Verifier moves towards the point where $\eta$ is to be taken.
So, we define
\begin{align}
\label{eq_skiphump}
 M(i) = &\ \{\pvar{\prophySkip{\rho(i),\rho'(i-1),v_\tau,v_\eta}{\tau,\eta}} \in P(i) \mid v_\tau \in V, \tau \in \Delta,\\ 
 & \quad \text{and $v_\eta$ and $\eta$ as in $\ch(j)$}\}.\nonumber
\end{align}

Lastly, if $\pvar{\step} \notin P(i)$, $\sh(c(i))>0$, and $\type{t(i)} = return$, then Verifier has reached the point where $\eta$ is to be taken.
So, we define the singleton set
\begin{align}
\label{eq_returnhump}
 M(i) = \set{\pvar{\prophyReturn{\rho(i),v_\eta}{\eta}} \in P(i) \text{ where $v_\eta$ and $\eta$ as in $\ch(j)$}}.
\end{align}
Again, this covers all subcases.

Using these definitions, if $M(i)$ is nonempty, we pick $\ch(i)$ arbitrarily from $M(i)$.
We define the move~$\rho'(i)$ of Verifier and the next transition~$\tau(i)$ for $r$ as uniquely determined by $\ch(i)$: in the first two cases (\texttt{CallHump} or \texttt{SkipHump}), we select the parameters~$v_\tau$ and $\tau$ of $\ch(i)$, in the latter case (\texttt{ReturnHump}), we select the parameters~$v_\eta$ and $\eta$ of $\ch(i)$.

If $M(i)$ is empty, let $M(i)=\emptyset$, $\rho'(i)$ be any $v \in \Succ{\rho'(i-1)}$, and let $\tau(i)$ be any transition of $\aut$ processing $(\lambda(\rho(i)),\lambda(\rho'(i)))$ that is enabled in $c(i)$. In both cases, we extend $r$ by $\tau(i)$ and the unique configuration reached by applying it in $c(i)$. In both cases, this is well-defined.

This concludes the definition of the strategy for Verifier.
\end{definition}

We now show that this strategy is winning, i.e., that each outcome of~$\gsgame(L(\aut,\tsys,\cal{P}))$ that is consistent with the strategy is in $L(\aut, \tsys, \cal{P})$.
Here, we argue about each outcome in isolation.
Thus, let us fix some notation.

Let $\combine{\alpha,\beta}$ be a fixed such outcome.
Let $\rho = \projone{\alpha}$ be the sequence of vertices picked by Falsifier and let $\rho'=\beta$ be the sequence of vertices picked by Verifier. Note that this is by construction of the strategy a path of $\tsys$.
Further, let $t = \lambda(\rho)$ and $t' = \lambda(\rho')$.
Finally, let $P(0)P(1)P(2)\cdots = \projtwo{\alpha}$ be the sequence of sets of prophecy variables picked by Falsifier.

As $\combine{\alpha,\beta}$ is an outcome consistent with our strategy, we have defined a sequence~$M(0)M(1)M(2)\cdots$ with $M(i) \subseteq P(i)$ for all $i$,
$\ch(i)$ for all $i$ with nonempty $M(i)$,
as well as a run~$r = c(0)\tau(0)c(1)\tau(0)\cdots$ processing $\combine{t,t'}$.

Recall that we want to show that $\combine{\alpha,\beta}$ is in $L(\aut, \tsys, \cal{P})$. 
There are some trivial cases: If the sequence~$\rho$ picked by Falsifier is not a path of $\tsys$, then $\combine{\alpha,\beta}$ is indeed in $L(\aut, \tsys, \cal{P})$. 
Similarly, if Falsifier makes a wrong prediction, i.e., if $\alpha$ does not satisfy
\[\forall i \in \nats.\ \forall P \in \mathcal{P}.\ \pvar{P} \in \projtwo{\alpha(i)} \leftrightarrow \lambda(\projone{\alpha(i)\alpha(i+1)\alpha(i+2)\cdots}) \in P,
\]
then $\combine{\alpha,\beta}$ is also in $L(\aut, \tsys, \cal{P})$.
Hence, we can focus on those outcomes in which Falsifier picks a path and where all of his predictions are correct.

\begin{assumption}
\label{assump_truthful}
Let $\combine{\alpha,\beta}$ be an outcome of $\gsgame(L(\aut,\tsys,\cal{P}))$ and let $\rho$ be defined as above.
We assume that $\rho$ is a path of $\tsys$ and that $\alpha$ satisfies
\[\forall i \in \nats.\ \forall P \in \mathcal{P}.\ \pvar{P} \in \projtwo{\alpha(i)} \leftrightarrow \lambda(\projone{\alpha(i)\alpha(i+1)\alpha(i+2)\cdots}) \in P.
\] 
\end{assumption}

Under this assumption, we show that Verifier has used the \nonempty-case at each position.
This is a crucial step in our proof that our strategy is winning, as it implies that the strategy always makes purposeful moves.

Before we begin, we need to formally define infixes of runs and the concatenation of runs.
Let $r = c_0\tau_0c_1\tau_1\cdots$ be a run. Then $r[m,n]$ denotes $c_m\tau_m\cdots\tau_nc_{n+1}$. The concatenation $r\cdot r'$ of two runs $r,r'$ is defined if $r = c_0\tau_0\cdots\tau_mc_{m+1}$ is a finite run and $r' = c'_0\tau'_0\cdots$ is a finite or infinite run with $c_{m+1} = c'_0$. The result $r\cdot r'$ is defined as expected.

Furthermore, let $c$ be a configuration of $\aut$, $v \in V$ a vertex of $\tsys$, and $t \in \Sigma^\omega$.
We define 
\[\opt{c,v,t} = \min_{r} \firstF{r},\]
where $r$ ranges over accepting runs of $\aut$ starting with configuration $c$ and processing $\combine{t,t'}$, where $t' \in \traces{\tsys_{\Succ{v}}}$.
If there is no such accepting run, then $\opt{c,v,t} = \infty$.

\begin{lemma}
\label{lemma_nonempty}
Let $\combine{\alpha,\beta}$ be an outcome of $\gsgame(L(\aut,\tsys,\cal{P}))$ that is consistent with the strategy constructed in \cref{def_strategy}
and that satisfies \cref{assump_truthful}. Then, Verifier has used the \nonempty-case at every position.
\end{lemma}

\begin{proof}
In the following, we write $t[i,\infty)$ for the suffix of $t$ starting at position~$i$ (which is included) and use all the notation~$(\rho$, $\rho'$, $t$, $t'$, the $P(i)$, the $M(i)$, the $\ch(i)$, and $r$) introduced above.

To begin with, we note that due to \cref{assump_truthful}, Falsifier picks a path of $\tsys$, i.e., $\projone{\alpha}$ is a path of $\tsys$, which implies $t \in \traces{\tsys}$.
Moreover, \cref{assump_truthful} guarantees that, for all $i\in \nats$ and all $P \in \cal{P}$, $t[i,\infty) \in P$ if and only if $\pvar{P} \in P(i)$, i.e., the prophecies are truthful.
This is a necessary precondition to prove that Verifier has always used the \nonempty-case, i.e., $M(i)$ is nonempty for each $i \geq 0$, because $M(i)$ is a subset of $P(i)$. So necessarily, $P(i)$ must be guaranteed to be correct (otherwise, Falsifier could just pick $P(i)=\emptyset$).
Note that the prophecies being truthful also implies that we can use them to check whether a position of $r$ is a step or a proper step. We will make use of that fact throughout the proof.

For each $i \geq 0$, we show inductively that $M(i)$ is nonempty.
Our induction hypothesis is that for each $j < i$, $M(j)$ is nonempty.

Case $i = 0$.
Since $\tsys \models \varphi$ and $t \in \traces{\tsys}$, there exists some $\tilde t \in \traces{\tsys}$ such that $\combine{t,\tilde t} \in L(\aut)$.
Thus, there exists a path~$\tilde\rho$ that starts from some $v_{\tau} \in V_\initmark$ and an accepting run~$\tilde r = c(0)\cdots$ on $\combine{t,\lambda(\tilde\rho)}$, say, with first transition~$\tau$, and with $\firstF{\tilde r} = \opt{(q_\initmark,\bot),\bullet,t}$.
Recall that $c(0) = (q_\initmark,\bot)$ and $0$ is a step because the stack is empty.
This yields
\[
t \in \begin{cases} 
    \prophyMCallAcc{\rho(0),\bullet,v_\tau,v_\eta}{\tau,\eta} & \text{if } \pvar{\properstep} \notin P(0) \text{ and } \type{t(0)} = call\\
    \prophyUCallAcc{\rho(0),\bullet,v_\tau}{\tau} & \text{if } \pvar{\properstep} \in P(0) \text{ and } \type{t(0)} = call\\
    \prophySkipAcc{\rho(0),\bullet,v_\tau}{\tau} & \text{if } \type{t(0)} = skip\\
    \prophyReturnAcc{\rho(0),\bullet,v_\tau}{\tau} & \text{if } \type{t(0)} = return
\end{cases}
\]
where $v_\eta$ (in $\tilde\rho$) and $\eta$ (in $\tilde r$) occur at the unique position in question w.r.t.\ the definition of\newline $\prophyMCallAcc{\rho(0),\bullet,v_\tau,v_\eta}{\tau,\eta}$.
That is, a matched call occurs, $\tau$ should be used to process the call, and $\eta$ should be used to process the matching return that occurs $n$ transitions later according to the definition.
Hence, the prophecy variable for some \myquote{\texttt{Step}}-prophecy is in $P(0)$.
By definition of $M(0)$, this variable is also selected to be in $M(0)$.
We conclude that $M(0)$ is nonempty.

Case $i > 0$.
If $i$ is a step, we show that there exists a path $\hat\rho$ and an accepting run $\hat r$ on $\combine{t[i,\infty),\lambda(\hat\rho)}$ with $\firstF{\hat r} = \opt{c(i),\rho'(i-1),t[i,\infty)}$ where $\hat\rho$ starts from some $v_{\tau} \in \Succ{\rho'(i-1)}$. By definition of the \myquote{\texttt{Step}}-prophecies, this shows that $t[i,\infty)$ is either in the {\texttt{MatchedCallStep}-}, or {\texttt{UnmatchedCallStep}-}, or {\texttt{SkipStep}-}, or {\texttt{UnmatchedReturnStep}}-prophecy, which is parameterized by the relevant vertices and transitions obtained from $\hat\rho$ and $\hat r$.
Which type of \myquote{\texttt{Step}}-prophecy it is depends on the type of $t(i)$ and whether $\pvar{\properstep} \in P(i)$.
Hence, the prophecy variable for some \myquote{\texttt{Step}}-prophecy is in $P(i)$.
By definition of $M(i)$, this variable is also selected to be in $M(i)$.
We can then conclude that $M(i)$ is nonempty which we need to show.

To obtain $\hat \rho$ and $\hat r$ as above, we make a further case distinction on whether $i-1$ is a step. 
We first assume that $i-1$ is a step.
By induction hypothesis, $M(i-1)$ is nonempty. 
Thus, $\ch(i-1)$ is of the form 
\begin{align*}
    & \prophyMCallAcc{\rho(i-1),\rho'(i-2),\rho'(i-1),v_{\eta}}{\tau(i-1),\eta} \text{ with $v_\eta \in V$, $\eta \in \Delta$, or}\\
    &\prophyUCallAcc{\rho(i-1),\rho'(i-2),\rho'(i-1)}{\tau(i-1)}, \text{ or}\\
    &\prophySkipAcc{\rho(i-1),\rho'(i-2),\rho'(i-1)}{\tau(i-1)}, \text{ or}\\
    &\prophyReturnAcc{\rho(i-1),\rho'(i-2),\rho'(i-1)}{\tau(i-1)},
\end{align*}
depending on the type of $t(i-1)$ and whether $\pvar{\properstep} \in P(i-1)$.
Here, and in the following, for the special case of $i=1$, $i-2$ is negative and we use $\rho'(-1) = \bullet$.

By definition of the \myquote{\texttt{Step}}-prophecies, this implies that there is some path $\tilde\rho = \rho'(i-1)\cdots$ and a run $\tilde r = c(i-1)\tau(i-1)\cdots$ on $\combine{t[i-1,\infty),\lambda(\tilde\rho)}$ that is accepting.
Clearly, $\tilde r[1,\infty]$ is an accepting run on $\combine{t[i,\infty),\lambda(\tilde\rho[1,\infty))}$ which is a necessary criterion for a prophecy variable for some \myquote{\texttt{Step}}-prophecy to be in $P(i)$ (and hence in $M(i)$).
However, $\tilde\rho[1,\infty)$ and $\tilde r[1,\infty]$ is simply a combination of a path and run which witnesses that Verifier can continue her path and continue building the run in a way that allows for acceptance.
Recall, as explained in the description of the prophecies, for Büchi acceptance it is actually necessary to not only have the possibility to be accepting, actual progress needs to be made towards visiting an accepting state.
Thus, instead of $\tilde\rho[1,\infty)$ and $\tilde r[1,\infty]$, we consider an \emph{optimal} combination of such a path and accepting run in terms of visiting an accepting state as soon as possible:
There is a path $\hat\rho$ that starts from some $v_\tau \in \Succ{\rho'(i-1)}$ and a run $\hat r = c(i)\cdots$ on $\combine{t[i,\infty),\lambda(\hat\rho)}$, say, with first transition $\tau$, and $\firstF{\hat r} = \opt{c(i),\rho'(i-1),t[i,\infty)}$.
This yields
\[
t[i,\infty) \in \begin{cases} 
    \prophyMCallAcc{\rho(i),\rho'(i-1),v_\tau,v_\eta}{\tau,\eta} & \text{if } \pvar{\properstep} \notin P(i) \text{ and } \type{t(i)} = call\\
    \prophyUCallAcc{\rho(i),\rho'(i-1),v_\tau}{\tau} & \text{if } \pvar{\properstep} \in P(i) \text{ and } \type{t(i)} = call\\
    \prophySkipAcc{\rho(i),\rho'(i-1),v_\tau}{\tau} & \text{if } \type{t(i)} = skip\\
    \prophyReturnAcc{\rho(i),v_\tau}{\tau} & \text{if } \type{t(i)} = return\\
\end{cases}
\]
where $v_\eta$ (in $\hat\rho$) and $\eta$ (in $\hat r$) occur at the unique position in question w.r.t.\ the definition of\newline $\prophyMCallAcc{\rho(i),\rho'(i-1),v_\tau,v_\eta}{\tau,\eta}$.
That is, a matched call occurs, $\tau$ should be used to process the call, and $\eta$ should be used to process the matching return that occurs $n$ transitions later according to the definition.
Hence, the prophecy variable for some \myquote{\texttt{Step}}-prophecy is in $P(i)$.
By definition of $M(i)$, this variable is also selected to be in $M(i)$.
We conclude that $M(i)$ is nonempty.

We now consider the case where $i-1$ is not a step.
Since $i$ is a step, this implies that the type of $t(i-1)$ is necessarily a return.
We are interested in the position where the matching call is made.
That is, the maximal $j < i-1$ where a call-transition is taken such that $\sh(c(j)) + 1 = \sh(c(i-1))$.
We have that $\sh(c(i)) = \sh(c(j)$, and the pair~$(j,i-1)$ of positions is a matching call-return pair.
Since $i$ is a step, $j$ is also a step.

In \cref{fig_pushdownstack}, such a situation is depicted choosing $i$ as position $8$, then the pair~$(2,7)$ of positions is the matching call-return pair.

In $j$ a matched call is made, hence $\pvar{\properstep} \notin P(j)$.
By induction hypothesis $M(j)$ is nonempty. Since $\pvar{\properstep} \notin P(j)$ and $\type{t(j)} = call$, we obtain that $\ch(j)$ is of the form 
\[
\prophyMCallAcc{\rho(j),\rho'(j-1),\rho'(j),v_\eta}{\tau(j),\eta} \text{ with $v_\eta \in V$ and $\eta\in\Delta$.}
\]
By definition of $\prophyMCallAcc{\rho(j),\rho'(j-1),\rho'(j),v_\eta}{\tau(j),\eta}$, this implies there exists some path $\tilde\rho = v_{j}v_{j+1}\cdots$ with $v_j = \rho'(j)$ and $v_i = v_\eta$ as well as an accepting run $\tilde r = c_j\tau_j\cdots$ on $\combine{t[j,\infty],\lambda(\tilde\rho)}$ with $c_j = c(j)$, $\tau_j = \tau(j)$ and $\tau_i = \eta$.
We now conclude that $M(i)$ is nonempty using suffixes of $\tilde\rho$ and $\tilde r$:
By induction hypothesis, $M(i-1)$ is nonempty.
By definition of $M(i-1)$, we obtain that $\rho'(i-1) = v_\eta$ and $\tau(i-1) = \eta$ because they are determined by $\ch(j)$.
Consequently, $c(i) = c_i$.
We split $\tilde\rho$ into $\rho'(j)v_{j+1}\cdots v_{i-1}v_\eta$ and its suffix $\tilde\rho'$.
We split $\tilde r$ into $c(j)\tau(j)c_{j+1}\cdots c_{i-1}\eta$ and its suffix $\tilde r' = c(i)\cdots$.
The run $\tilde r'$ on $\combine{t[i,\infty),\lambda(\tilde\rho')}$ is accepting.
As explained before (in the case where $i-1$ is assumed to be a step), the path $\tilde\rho'$ and run $\tilde r'$ are a combination that witnesses that Verifier can continue in a way that will be accepting, but we are looking to make the most progress towards accepting.
Hence, we take a combination of path and accepting run which are optimal in that sense, i.e., where an accepting state is visited as soon as possible:
There exists a path $\hat\rho$ that starts from some $v_\tau \in \Succ{v_\eta}$ (recall $v_\eta = \rho'(i-1)$) and an accepting run $\hat r= c(i)\cdots$, say with first transition $\tau$, which is optimal, i.e., $\firstF{\hat r} = \opt{c(i),\rho'(i-1),t[i,\infty)}$.
This yields
\[
t[i,\infty) \in \begin{cases}
    \prophyMCallAcc{\rho(i),\rho'(i-1),v_\tau,v_\eta}{\tau,\eta} &  \text{if } \pvar{\properstep} \notin P(i) \text{ and } \type{t(i)} = call\\
    \prophyUCallAcc{\rho(i),\rho'(i-1),v_\tau}{\tau} & \text{if } \pvar{\properstep} \in P(i)\text{ and } \type{t(i)} = call\\
    \prophySkipAcc{\rho(i),\rho'(i-1),v_\tau}{\tau} & \text{if } \type{t(i)} = skip\\
    \prophyReturnAcc{\rho(i),v_\tau}{\tau} & \text{if } \type{t(i)} = return\\
\end{cases}
\]
where $v_\eta$ (in $\hat\rho$) and $\eta$ (in $\hat r$) occur at the unique position in question w.r.t.\ the definition of \newline$\prophyMCallAcc{\rho(i),\rho'(i-1),v_\tau,v_\eta}{\tau,\eta}$.
That is, a matched call occurs, $\tau$ should be used to process the call, and $\eta$ should be used to process the matching return that occurs $n$ transitions later according to the definition.
Hence, the prophecy variable for some \myquote{\texttt{Step}}-prophecy is in $P(i)$.
By definition of $M(i)$, this variable is also selected to be in $M(i)$.
We conclude that $M(i)$ is nonempty.

We have completed the case where $i$ is a step.
Now we consider the case where $i$ is not a step, i.e., $i$ is in a hump.
Consequently, for this part of the proof, also \myquote{\texttt{Hump}}-prophecies are relevant.
We are interested in the previous call, that is, the maximal $j < i$ where a call-transition is taken such that $\sh(c(j)) + 1 = \sh(c(i))$.
By induction hypothesis, $M(j)$ is nonempty, thus, $\ch(j)$ is of the form 
\[
\prophyMCallAcc{\rho(j),\rho'(j-1),\rho'(j),v_\eta}{\tau(j),\eta} \text{ with $v_\eta \in V$ and $\eta\in\Delta$}
\]
or
\[
\prophyCall{\rho(j),\rho'(j-1),\rho'(j),v_\eta,v_\vartheta}{\tau(j),\eta,\vartheta} \text{ with $v_\eta,v_\vartheta \in V$ and $\eta,\vartheta \in \Delta$.}
\]

In \cref{fig_pushdownstack}, the first situation occurs by taking $i$ as position~$5$, then $j$ is position~$2$.
The second situation occurs by taking $i$ as $6$, then $j$ is $5$.

By definition of $\prophyMCallAcc{\cdot}{\cdot}$ resp.\ $\prophyCall{\cdot}{\cdot}$, the transition $\tau(j)$ processes the call $(t(j),\lambda(\rho'(j))) \in \Sigma^2$ such that there is a transition $\eta$ that processes the matching return $(t(j+k),\lambda(v_\eta)) \in \Sigma^2$ for some vertex $v_\eta \in V$ which occurs, say, $k$ transitions later.
Having this prophecy as $\ch(j)$, we show that Verifier plays in a way that honors the commitment made for the return.
We have that $j < i \leq k$.

We distinguish how $i$ is reached, i.e., we take a look at $i-1$.
By induction hypothesis, $M(i-1)$ is nonempty.

Firstly, assume the type of $t(i-1)$ is a skip as this is the simplest case.
Thus, $\ch(i-1)$ is of the form 
\[
\prophySkip{\rho(i-1),\rho'(i-2),\rho'(i-1),v_\eta}{\tau(i-1),\eta} \text{ with $v_\eta$ and $\eta$ as in $\ch(j)$.}
\]
Recall the case of \cref{def_strategy} where the set $M$ is populated by prophecy variables for \myquote{\texttt{SkipHump}}-prophecies (see \cref{eq_skiphump}), to see that $v_\eta$ and $\eta$ are determined by $\ch(j)$.

For the unique suitable $n$ (that is where the next return occurs according to the definition of $\prophySkip{\rho(i-1),\rho'(i-2),\rho'(i-1),v_\eta}{\tau(i-1),\eta}$), let $\tilde\rho$ be a finite path of the form
\[
 \tilde\rho = v_{0}\cdots v_n \text{ with } v_0 = \rho'(i-1) \text{ and } v_n= v_\eta
\]
and $\tilde r$ be a finite run on $\combine{t[i-1,n],\lambda(\tilde \rho)}$ of the form
\[
\tilde r = c_0\tau_0\cdots c_n\tau_n c_{n+1} \text{ with $c_0 = c(i-1)$ and $\tau_0 = \tau(i-1)$ (and hence $c_1 = c(i)$ and $\tau_n = \eta$)},
\]
such that $\tilde \rho$ and $\tilde r$ witness that the prophecy variable for $\ch(i-1)$ is in $P(i-1)$ (hence also in $M(i-1)$).
Taking $\tilde \rho[1,n] = v_1\cdots v_n$ and $\tilde r[1,n] = c_1\tau_1\cdots c_n\tau_n c_{n+1}$ gives us a path starting in $\Succ{\rho'(i-1)}$ and reaching $v_\eta$ as well as a run starting from $c(i)$ using $\eta$ as its last transition.
As before, we are not interested in any combination of path and run (such as $\tilde \rho[1,n]$ and $\tilde r[1,n]$) that reaches $v_\eta$ where $\eta$ can be taken in $n$ transitions, we are looking for an optimal combination of path and run with these properties in terms of visiting an accepting state as soon as possible:
There is a path 
\[
\hat \rho = v'_1\cdots v'_n \text{ with $v'_1 \in \Succ{\rho'(i-1)}$ and $v'_n = v_\eta$, and}
\]
a run on $\combine{t[i,n],\lambda(\hat \rho)}$ 
\[
 \hat r = c'_1\tau'_1\cdots c'_n\tau'_n c'_{n+1} \text { with } c'_1 = c(i) \text{ and } \tau'_n =\eta \text{ (and hence } c'_{n+1} = c_{n+1}) 
\]
such that for all paths $\hat \rho' = v''_0v''_1\cdots v''_{n-1}v_\eta$ where $v''_0 \in \Succ{\rho'(i-1)}$ and runs $\hat r' = c(i)\cdots \eta c_{n+1}$ on $\combine{t[i,n],\lambda(\hat \rho')}$ we have $\firstF{\hat r} \leq \firstF{\hat r'}$.
This yields
\[
t[i,\infty) \in \begin{cases} 
    \prophyCall{\rho(i),\rho'(i-1),v'_1,v'_k,v_\eta}{\tau'_1,\tau'_k,\eta} & \text{if $\type{t(i)} = call$}\\
    \prophySkip{\rho(i),\rho'(i-1),v'_1,v_\eta}{\tau'_1,\eta} & \text{if $\type{t(i)} = skip$}\\
    \prophyReturn{\rho(i),v_\eta}{\eta} & \text{if $\type{t(i)} = return$},
\end{cases}
\]
where $v'_k$ (in $\hat \rho$) and $\tau'_k$ (in $\hat r)$ occur at the unique position in question w.r.t.\ the definition of \newline$\prophyCall{\rho(i),\rho'(i-1),v'_1,v'_k,v_\eta}{\tau'_1,\tau'_k,\eta}$.
As a reminder of its definition, a matched call occurs, $\tau'_1$ processes the call, and $\tau'_k$ processes the matching return.
Hence, the prophecy variable for some \myquote{\texttt{Hump}}-prophecy is in $P(i)$.
By definition of $M(i)$, this variable is also selected to be in $M(i)$.
We conclude that $M(i)$ is nonempty.

Secondly, assume the type of $t(i-1)$ is a call.
This is the situation depicted in \cref{fig_pushdownstack} by taking $i$ as $6$, then $j$ is $5$.
Then, $j = i-1$, and we re-express $\ch(j)$ as $\ch(i-1)$ by replacing $j$ with $i-1$:
\[
\prophyCall{\rho(i-1),\rho'(i-2),\rho'(i-1),v_\eta,v_\vartheta}{\tau(i-1),\eta,\vartheta} \text{ with $v_\eta,v_\vartheta \in V$ and $\eta,\vartheta \in \Delta$.}
\]
Again, we take a look at a path and run that witness that the prophecy variable for $\ch(i-1)$ is in $P(i-1)$ (hence also in $M(i-1)$).
For the unique suitable $m < n$ (as in the definition of $\prophyCall{\cdot}{\cdot}$ three lines above), let $\tilde\rho$ be a finite path of the form
\[
 \tilde\rho = v_{0}\cdots v_n \text{ with  $v_0 = \rho'(i-1)$, $v_m = v_\eta$, and $v_n = v_\vartheta$},
\]
and $\tilde r$ be a finite run $\combine{t[i-1,n],\lambda(\tilde \rho)}$ of the form
\[
\tilde r = c_0\tau_0\cdots c_n\tau_n c_{n+1} \text{ with $c_0 = c(i-1)$, $\tau_0 = \tau(i-1)$, hence $c_1 = c(i)$, $\tau_m = \eta$, and $\tau_n = \vartheta$}
\]
such that $\tilde \rho$ and $\tilde r$ witness that the prophecy variable for $\ch(i-1)$ is in $P(i)$ (hence in $M(i)$).
Taking the infix~$\tilde \rho[1,m] = v_1\cdots v_m$ and $\tilde r[1,m] = c_1\tau_1\cdots c_m\tau_n c_{m+1}$ gives us a path starting in $\Succ{\rho'(i-1)}$ and reaching $v_\eta$ as well as a run starting from $c(i)$ using $\eta$ as its last transition.
As before, we are considering an optimal (in terms of visiting an accepting state as soon as possible) combination where the path and run end with the same vertex respectively transition as $\tilde{\rho}$ and $\tilde{r}$: 
There is a path 
\[
\hat \rho = v'_1\cdots v'_m \text{ with  $v'_1 \in \Succ{\rho'(i-1)}$  and $v'_m = v_\eta$, and }
\]
a run on $\combine{t[i,n],\lambda(\hat \rho)}$ 
\[
 \hat r = c'_1\tau'_1\cdots c'_m\tau'_m c'_{m+1} \text{ with  $c'_1 = c(i)$  and $\tau'_m =\eta$ (and hence $c'_{m+1} = c_{m+1}$)} 
\]
such that for all paths $\hat \rho' = v''_0v''_1\cdots v''_{m-1}v_\eta$ where $v''_0 \in \Succ{\rho'(i-1)}$ and runs $\hat r' = c(i)\cdots \eta c_{m+1}$ on $\combine{t[i,n],\lambda(\hat \rho')}$ we have $\firstF{\hat r} \leq \firstF{\hat r'}$.
This yields
\[
t[i,\infty) \in \begin{cases} 
    \prophyCall{\rho(i),\rho'(i-1),v'_1,v'_k,v_\eta}{\tau'_1,\tau'_k,\eta} & \text{if $\type{t(i)} = call$}\\
    \prophySkip{\rho(i),\rho'(i-1),v'_1,v_\eta}{\tau'_1,\eta} & \text{if $\type{t(i)} = skip$}\\
    \prophyReturn{\rho(i),v_\eta}{\eta} & \text{if $\type{t(i)} = return$},
\end{cases}
\]
where $v'_k$ (in $\hat \rho$) and $\tau'_k$ (in $\hat r)$ occur at the unique position in question w.r.t.\ the definition of \newline$\prophyCall{\rho(i),\rho'(i-1),v'_1,v'_k,v_\eta}{\tau'_1,\tau'_k,\eta}$.
As a reminder of its definition, a matched call occurs, $\tau'_1$ should be used to process the call, and $\tau'_k$ should be used to process the matching return.
Hence, the prophecy variable for some \myquote{\texttt{Hump}}-prophecy is in $P(i)$.
By definition of $M(i)$, this variable is also selected to be in $M(i)$.
We conclude that $M(i)$ is nonempty.

Lastly, assume the type of $t(i-1)$ is a return.
Such a situation is depicted in \cref{fig_pushdownstack} by taking $i$ as $5$, as position~$5$ lies in between the matching call-return pair~$(2,7)$ (hence, $j$ is $2$).

We are interested in its matching call:
There exists a maximal $j' < i-1$ where a call-transition is taken such that $\sh(c(j')) + 1 = \sh(c(i-1))$.
Note that $j'$ is not a step, as the call is undone by the return in $i-1$.
Furthermore, note that we have $j < j'$ and the type of $t(\ell)$ is a skip for every $j < \ell < j'$.

Returning to the situation depicted in \cref{fig_pushdownstack} by taking $i$ as $5$ where $i-1$ is a return, we obtain that $j'$ is position $3$.

By induction hypothesis, $M(j')$ is nonempty. Thus, $\ch(j')$ is of the form 
\[
\prophyCall{\rho(j'),\rho'(j'-1),\rho'(j'),v_{\eta'},v_\eta}{\tau(j'),\eta',\eta}
\]
for some $v_{\eta'} \in V$ and some return-transitions $\eta'$ and $v_\eta$ and $\eta$ as in $\ch(j)$.
Recall the case of \cref{def_strategy} where the set $M$ is populated by prophecy variables for \myquote{\texttt{CallHump}}-prophecies (see \cref{eq_callhump}), to see that $v_\eta$ and $\eta$ are determined by $\ch(j)$.

By induction hypothesis, $M(i-1)$ is nonempty. Thus, $\ch(i-1)$ is of the form 
\[
\prophyReturn{\rho(i-1),\rho'(i-1)}{\tau(i-1)}
\]
where $\rho'(i-1) = v_{\eta'}$ and $\tau(i-1) = \eta'$ with $v_{\eta'}$ and $\eta'$ as in $\ch(j')$.
Recall the case of \cref{def_strategy} where the set $M$ is populated by prophecy variables for \myquote{\texttt{ReturnHump}}-prophecies (see \cref{eq_returnhump}), to see that $v_{\eta'}$ and $\eta'$ are determined by $\ch(j')$.

There are a path and run that witness that the prophecy variable for $\ch(j')$ is in $P(j')$ (hence in $M(j')$).
In combination with $\ch(i-1)$, we will obtain a path and run which yields that $M(i)$ is nonempty.

For the unique suitable $m < n$ as in the definition of 
\[
\prophyCall{\rho(j'),\rho'(j'-1),\rho'(j'),v_{\eta'},v_\eta}{\tau(j'),\eta',\eta},
\]
let $\tilde\rho$ be a finite path of the form
\[
 \tilde\rho = v_{0}\cdots v_n \text{ with  $v_0 = \rho'(j')$, $v_m = v_{\eta'}$, and $v_n = v_\eta$},
\]
and $\tilde r$ be a finite run $\combine{t[j',n],\lambda(\tilde \rho)}$ of the form
\[
\tilde r = c_0\tau_0\cdots c_n\tau_n c_{n+1} \text{ with $c_0 = c(j')$, $\tau_0 = \tau(j')$, $\tau_m = \eta'$, and $\tau_n = \eta$}
\]
such that $\tilde \rho$ and $\tilde r$ witness that the prophecy variable for $\ch(j')$ is in $P(j')$ (hence in $M(j')$).
Using $\ch(i-1)$, we know that $v_m = \rho'(i-1)$ and $\tau_m = \tau(i-1)$.
Since $c_0 = c(j)$, this yields also $c_{m+1} = c(i)$.
Taking the suffixes $\tilde\rho[m+1,n]$ and $\tilde r[m+1,n]$ gives us a path starting in $\Succ{\rho'(i-1)}$ reaching $v_\eta$ as well as a run starting from $c(i)$ using $\eta$ as its last transition.
As before, instead of considering the combination $\tilde\rho[m+1,n]$ and $\tilde r[m+1,n]$, we consider an optimal combination (in terms of visiting an accepting state as soon as possible) where the path and run end with the same vertex respectively transition as $\tilde\rho[m+1,n]$ and $\tilde r[m+1,n]$:
There is a path 
\[
\hat \rho = v'_{m+1}\cdots v'_n \text{ with  $v'_{m+1} \in \Succ{\rho'(i-1)}$  and $v'_n = v_\eta$, and }
\]
a run on $\combine{t[i,n],\lambda(\hat \rho)}$ 
\[
 \hat r = c'_{m+1}\tau'_1\cdots c'_n\tau'_n c'_{n+1} \text{ with  $c'_{m+1} = c(i)$  and $\tau'_n =\eta$ (and hence $c'_{n+1} = c_{n+1}$)} 
\]
such that for all paths $\hat \rho' = v''_{m+1}v''_{m+2}\cdots v''_{n-1}v_\eta$ where $v''_{m+1} \in \Succ{\rho'(i-1)}$ and runs $\hat r' = c(i)\cdots \eta c_{m+1}$ on $\combine{t[i,n],\lambda(\hat \rho')}$ we have $\firstF{\hat r} \leq \firstF{\hat r'}$.
This yields
\[
t[i,\infty) \in \begin{cases} 
    \prophyCall{\rho(i),\rho'(i-1),v'_{m+1},v'_k,v_\eta}{\tau'_{m+1},\tau'_k,\eta} & \text{if $\type{t(i)} = call$}\\
    \prophySkip{\rho(i),\rho'(i-1),v'_{m+1},v_\eta}{\tau'_{m+1},\eta} & \text{if $\type{t(i)} = skip$}\\
    \prophyReturn{\rho(i),v_\eta}{\eta} & \text{if $\type{t(i)} = return$},
\end{cases}
\]
where $v'_k$ (in $\hat \rho$) and $\tau'_k$ (in $\hat r)$ occur at the unique position in question w.r.t.\ the definition of \newline$\prophyCall{\rho(i),\rho'(i-1),v'_{m+1},v'_k,v_\eta}{\tau'_{m+1},\tau'_k,\eta}$.
As a reminder of its definition, a matched call occurs, $\tau'_{m+1}$ should be used to process the call, and $\tau'_k$ should be used to process the matching return.
Hence, the prophecy variable for some \myquote{\texttt{Hump}}-prophecy is in $P(i)$.
By definition of $M(i)$, this variable is also selected to be in $M(i)$.
We conclude that $M(i)$ is nonempty.
\end{proof}

We now show that since Verifier has always used the \nonempty-case (\cref{lemma_nonempty}) for all outcomes that satisfy \cref{assump_truthful}, the run~$r$ defined during the strategy definition is accepting.

\begin{lemma}
\label{lemma_runAcc}
Let $\combine{\alpha,\beta}$ be an outcome of the Gale-Stewart game $\gsgame(L(\aut,\tsys,\cal{P}))$ that is consistent with the strategy constructed in \cref{def_strategy} and that satisfies \cref{assump_truthful}.
Let $t,t'$ be the traces induced by $\combine{\alpha,\beta}$.
The run~$r$ of $\aut$ on $\combine{t,t'}$ constructed during the strategy definition is accepting.
\end{lemma}

\begin{proof}
We first remark that \cref{assump_truthful} guarantees that $t \in \traces{\tsys}$, this is a necessary precondition for the run $r$ on $\combine{t,t'}$ to be accepting (recall we also assume that $\tsys \models \varphi$).
As above, we use all the necessary notation introduced above \cref{lemma_nonempty}.
Furthermore, \cref{lemma_nonempty} implies that $\ch(i)$ is defined for all $i$.

The run $r = c(0)\tau(0)\cdots$ is constructed based on the prophecies.  
In general, to pick $\tau(i)$, the prophecy $\ch(i)$ indicates how to extend $c(0)\tau(0)\cdots\tau(i-1)c(i)$ in a way such that an accepting state is visited as soon as possible. 
But there is an important distinction to be made whether $i$ is a step or not.
If $i$ is step, the run $r$ is in a situation where all open calls have been matched (or never will be, hence are not important).
The prophecy $\ch(i)$ simply indicates an accepting run starting from (a sufficient representation of) $c(i)$ that visits an accepting state as soon as possible.
Now, if $i$ is not step, then $i$ is in a hump, i.e., there are open calls that will be matched in the future.
In that case, the prophecy $\ch(i)$ is parameterized by how the latest open call is to be closed.
Thus, the finite run indicated by $\ch(i)$ is a run that respects how the call is to be closed and furthermore this run visits an accepting state (if possible at all) as soon as possible. 

We prove that this construction implies that $r$ is build in a way such that whenever $r$ is at a step, it follows an accepting run that visits an accepting state as soon as possible until the accepting state has been visited.
When the accepting state has been reached, some calls in $r$ are still open, but the commitment on how to close them has already been made.
The run $r$ continues and closes them.
All calls have been closed when the next step is reached.
As every run contains infinitely many steps, this will eventually happen.
Thus, the run $r$ is free again to follow its way to an accepting state as soon as possible.
A this happens infinitely many times, an accepting state is visited infinitely many times. It follows that the run~$r$ is indeed accepting.

Formally, we show for all $i \geq 0$ that if $i$ is a step then $\opt{c(i),\rho'(i-1),t[i,\infty)} < \infty$ and furthermore, there exists a path
\[
\tilde \rho = v_iv_{i+1}\cdots \text{ with $v_i \in \Succ{\rho'(i-1)}$}, 
\]
and a run 
\[
\tilde r = c_i\tau_i\cdots \text{ with $c_i = c(i)$}
\]
on $\combine{t[i,\infty),\lambda(\tilde \rho)}$ such that $\firstF{\tilde r} = \opt{c(i),\rho'(i-1),t[i,\infty)}$ and
\[
\rho'(\ell) = v_\ell,\ c(\ell) = c_\ell,\ \tau(\ell) = \tau_\ell \text{ (recall that $r = c(0)\tau(0)\cdots c(\ell)\tau(\ell)\cdots$)}
\]
for all $i \leq \ell \leq i+k$, where $k = \firstF{\tilde r}$.
In short, we show that the optimal run $\tilde r$ and the run $r$ constructed during the strategy construction now coincide for the next $k$ positions.
Since at the $k$-th position (starting from $i$ which is a step), $\tilde r$ visits an accepting state, so does $r$.
Since $r$ has infinity many steps, this yields that at infinitely many positions $F$ is visited, i.e., $r$ is accepting.

Fix any $i$ that is a step. Then $\ch(i)$ is of the form
\begin{align*}
&\prophyMCallAcc{\rho(i),\rho'(i-1),\rho'(i),v_\eta}{\tau(i),\eta} \text{ if $\pvar{\properstep} \notin P(i)$ and $\type{t(i)} = call$}\\
&\prophyUCallAcc{\rho(i),\rho'(i-1),\rho'(i)}{\tau(i)}\text{ if $\pvar{\properstep} \in P(i)$ and $\type{t(i)} = call$}\\
&\prophySkipAcc{\rho(i),\rho'(i-1),\rho'(i)}{\tau(i)}\text{ if $\type{t(i)} = skip$}\\
&\prophyReturnAcc{\rho(i),\rho'(i)}{\tau(i)}\text{ if $\type{t(i)} = return$}
\end{align*}
for some vertex $v_\eta$ and some $\Sigmareturn$-transition $\eta$ in first case.
By definition of the prophecy~$\ch(i)$, this directly yields that $\opt{c(i),\rho'(i-1),t[i,\infty)} < \infty$ and there exists a path $\tilde\rho$ and a run $\tilde r$ on $\combine{t[i,\infty),\lambda(\tilde \rho)}$ such that 
\begin{equation}
    \label{eq_opti}
    \firstF{\tilde r} = \opt{c(i),\rho'(i-1),t[i,\infty)}.
\end{equation}

Furthermore, $v_i = \rho'(i) \in \Succ{\rho'(i-1)}$ and $c_i = c(i)$.
We note that there can actually be several of these optimal combinations, let $S$ be the set of these optimal $(\tilde \rho,\tilde r)$.

We argue that $\rho'[i,\infty)$ and $r[i,\infty)$ follows some optimal $\tilde\rho$ and $\tilde r$ at least until $F$ is visited.
Formally, we prove that there exists  $(\tilde \rho,\tilde r) \in S$, say $\tilde \rho = v_iv_{i+1}\cdots$ and $\tilde r = c_i\tau_i\cdots$, such that for all $i \leq \ell \leq i+k$ where $k = \firstF{\tilde r}$, $\ch(\ell)$ is picked w.r.t.\ $\tilde \rho$ and $\tilde r$ meaning the following:
\begin{description}
    \item[Condition 1]
    \label{item1}
    We have $\rho'(\ell) = v_\ell,\ c(\ell) = c_\ell,\ \tau(\ell) = \tau_\ell$. 
    \item[Condition 2]
    \label{item2}If the type of $t(\ell)$ is a call which has a matching return later, then $\ch(\ell)$ fixes the return (say $n$ positions later) as $\tau_{\ell+n}$ (from $\tilde r$) to be taken at $v_{\ell+n}$ (from $\tilde \rho$).
    Thus, \cref{def_strategy} ensures that $\rho'(\ell+n)$ will be $v_{\ell+n}$ and $\tau(\ell + n)$ will be $\tau_{\ell+n}$.
\end{description}
We note the first condition covers the statement given at the start of this lemma which is enough to show that the constructed run is accepting.
The second condition is only needed for correctness.

Before we start with the proof, we note that, for all $(\tilde \rho,\tilde r) \in S$, the prefix $\tilde r[0,k)$ does not visit $F$, and since $\tilde r$ is optimal, we obtain that 
\begin{align}
\label{eq_1}
   \firstF{\tilde r[i-\ell,\infty)} = \opt{c(\ell),\rho'(\ell-1),t[\ell,\infty)},\text{ and } 
\end{align}
\begin{align}
\label{eq_2}
\firstF{\tilde r} = k = (\ell - i) + \firstF{\tilde r[\ell-i,\infty)}
\end{align}
for all $i \leq \ell \leq i+k$.

For any $\ell$ with $i < \ell \leq i+k$, we assume that $\ch(i)$ up to $\ch(\ell-1)$ picked w.r.t.\ $\tilde \rho$ and $\tilde r$ for some $(\tilde \rho,\tilde r) \in S$.
We show that $\ch(\ell)$ is picked w.r.t.\ $\tilde \rho$ and $\tilde r$ for some $(\tilde \rho,\tilde r) \in S'$, where $S' \subseteq S$ is the subset of pairs $(\tilde \rho,\tilde r)$ which are compatible with $\ch(i)$ up to $\ch(\ell-1)$.
Now, towards a contradiction, assume $\ch(\ell)$ is not compatible with any $(\tilde \rho,\tilde r) \in S'$.
The prophecy $\ch(\ell)$ implies that there is a path $\hat \rho = \rho'(\ell)\cdots$ and run $\hat r = c(\ell)\tau(\ell)\cdots$ such that $\firstF{\hat r}$ is minimal among relevant alternative paths/runs.
Which paths/runs are relevant, is determined by whether $\ell$ is a step.
\begin{itemize}
    \item $\ell$ is a step: $\ch(\ell)$ is a \myquote{\texttt{Step}}-prophecy, these types of prophecies set no constraints on the future of the path/run that have to be respected.
    \item $\ell$ is not a step, i.e., $\ell$ is in a hump:  $\ch(\ell)$ is a \myquote{\texttt{Hump}}-prophecy, these types of prophecies set constraints on the future of the path/run regarding upcoming returns that have to be respected.
\end{itemize} 
See the definition of the individual prophecies for a more formal understanding. 
We now use the path $\hat \rho$ and run $\hat r$ which are an optimal combination among relevant alternatives to arrive at a contradiction.

For all $(\tilde \rho,\tilde r) \in S'$, it holds that $\rho'[i,\ell)$ is a prefix of $\tilde \rho$ and $r[i,\ell) = c(i)\tau(i)\cdots c(\ell)$ is a prefix of $\tilde r$.
We show that $(\rho'[i,\ell)\cdot\hat\rho,r[i,\ell)\cdot\hat r) \in S'$.
Take any $(\tilde \rho,\tilde r) \in S'$. We know that
\begin{equation}
\label{eq_3}
\firstF{\hat r} \geq \opt{c(\ell),\rho'(\ell-1),t[\ell,\infty)} = \firstF{\tilde r[i-\ell,\infty)}
\end{equation}
by definition of $\opt{\cdot,\cdot,\cdot}$ and \cref{eq_1}.
By \cref{def_strategy}, $\ch(\ell)$ continues the path $\rho'$ and the run $r$ in a way that respects the previous choices about upcoming returns. Thus, in particular, choices made by $\ch(i)$ up to $\ch(\ell-1)$.
We recall that we are guaranteed by Condition 2, that the choices about upcoming returns are made w.r.t.\ $\tilde \rho$ and $\tilde r$ for all $(\tilde \rho,\tilde r) \in S'$. 
Hence, for $\ch(\ell)$, in order to determine that $\hat \rho$ and $\hat r$ yield a combination such that $\firstF{\hat r}$ is minimal among all other relevant path/run combinations, all $\tilde \rho[i-\ell,\infty)$ and $\tilde r[i-\ell,\infty)$ with $(\tilde \rho,\tilde r) \in S'$ have been considered.
We can conclude that
\begin{equation}
\label{eq_4}
\firstF{\hat r} \leq \firstF{\tilde r[i-\ell,\infty)},
\end{equation}
for all $(\tilde \rho,\tilde r) \in S'$.
Thus, combining \cref{eq_3} and \cref{eq_4}, we have 
\begin{equation}
\label{eq_5}
\firstF{\hat r} = \firstF{\tilde r[i-\ell,\infty)},
\end{equation}
for all $(\tilde \rho,\tilde r) \in S'$.
Using \cref{eq_2} and \cref{eq_5}, we obtain
\begin{equation}
\label{eq_6}
\firstF{\tilde r} = (\ell - i) + \firstF{\hat r}
\end{equation}
for all $(\tilde \rho,\tilde r) \in S'$.
Finally, using that $(\ell-i) + \firstF{\hat r} = \firstF{r[i,\ell) \cdot \hat r}$ and \cref{eq_opti}, yields
\begin{equation}
\label{eq_7}
\firstF{r[i,\ell) \cdot \hat r} = \opt{c(i),\rho'(i-1),t[i,\infty)}.
\end{equation}
Thus, $(\rho'[i,\ell)\cdot\hat\rho,r[i,\ell)\cdot\hat r) \in S'$, which contradicts the initial assumption that $\ch(\ell)$ is not compatible to any pair in $S'$.
Hence, $\ch(\ell)$ is compatible to at least one pair in $S'$.

Recall that compatible means it satisfies Conditions 1 and 2, so we have proven that also in the $\ell$-th position after a step, the run $r$ has copied an optimal run $\tilde r$ starting from $c(i)$ processing $t[i,\infty)$ (this is Condition 1).
Since $\ell$ ranges between $0$ and $k$ (inclusive), and $\tilde r$ visits an accepting state at position $k$, we have shown that $r$ visits an accepting state at position $i+k$.

As this is true for all positions $i$ in $r$ that are a step, and as there are infinitely many steps in $r$, we have proven that $r$ visits infinitely many accepting states, i.e., it is accepting.
\end{proof}

We have just shown that if $\combine{\alpha,\beta}$ is an outcome of the Gale-Stewart game~$\gsgame(L(\aut,\tsys,\cal{P}))$ where Verifier played according to the strategy proposed in \cref{def_strategy}, then there exists an accepting run $r$ of $\aut$ (recall, that is, the automaton of the original formula $\varphi$) on $\combine{t,t'}$ where $t,t'$ are the traces induced by the outcome of the game.
Hence, our characterization is indeed complete: if $\tsys \models \forall \pi.\ \exists \pi'.\ \aut$ then Verifier wins $\gsgame(L(\aut,\tsys, \cal{P}))$

\section{Conclusion}

We have introduced \hylogicpda and \hylogicvpa to specify context-free hyperproperties, which extend $\omega$-PDA and $\omega$-VPA, respectively, by quantification over traces, just like \hyltl extends \ltl by trace quantification. 
Not surprisingly, \hylogicpda model-checking is undecidable for all quantifier-fragments but $\Sigma_1$ (i.e., $\exists^*$ formulas), as the undecidable universality problem for PDA can be encoded using formulas of the form~$\forall \pi_0.\ \aut$.

For \hylogicvpa, the situation is better. Here, model-checking for $\Pi_{n,1}$ and $\Sigma_{n,1}$ formulas (formulas with arbitrary quantifier alternations where the partition into calls, returns, and skips only depends on the letters of the traces quantified in the first quantifier block) is decidable. This covers many hyperproperties from the literature.
We complemented this by showing that model-checking for $\Pi_{2,2}$ and $\Sigma_{2,2}$ (i.e., one quantifier alternation, but the partition depends only on the second block of quantifiers) is undecidable. 
These results also immediately imply undecidability for all $\Sigma_{n,j}$ and $\Pi_{n,j}$  for $n>2$ and $j>1$.

Furthermore, we presented a game-based characterization of \hyltl model-checking for formulas with one quantifier alternation using prophecies. This yields explanations and counterexamples for such formulas.
In further work, we investigate whether this approach can be lifted to arbitrary quantifier alternations.

\bibliographystyle{plain}
\bibliography{bib}

\end{document}